 \documentclass[final,3p,times,authoryear]{elsarticle}

\usepackage{amssymb}
\usepackage{algorithmic}
 \usepackage{amsthm}
\usepackage{graphicx}
\usepackage{caption}
\usepackage{subcaption} 

\usepackage{lipsum,graphicx,multicol}
\usepackage{mathtools}
\usepackage{blkarray}
\usepackage{multirow}
\usepackage{hyperref}
\usepackage{makecell}
\hypersetup{
     colorlinks   = true,
     citecolor    = blue
}

\journal{working paper at arXiv.org}

\begin{document}

\begin{frontmatter}



\title{Measuring Systemic Risk:\\ Robust Ranking Techniques Approach}


\author{Amirhossein Sadoghi}

\address{Frankfurt School of Finance \& Management,\\
 Frankfurt am Main E.mail: a.sadoghi@fs.de}

\begin{abstract}
In this research, we introduce a robust metric to identify Systemically Important Financial Institution (SIFI) in a financial network by taking into account both common idiosyncratic shocks and contagion through counterparty exposures. We develop an efficient algorithm to rank financial institutions by formulating a fixed point problem and reducing it to a non-smooth convex optimization problem. We then study the underlying distribution of the proposed metric and analyze the performance of the algorithm by using different financial network structures. Overall, our findings suggest that the level of interconnection and position of institutions in the financial network are important elements to measure systemic risk and identify SIFIs. Results show that increasing the levels of out- and in-degree connections of an institution can have a diverse impact on its systemic ranking. Additionally, on the empirical side, we investigate the factors which lead to the identification of Global Systemic Important Banks (G-SIB) by using a panel dataset of the largest banks in each country. Our empirical results supports the main findings of the theoretical model.
\end{abstract}


\footnotetext[1]{Corresponding Author, address: Sonnemannstraße 9-11 60314 Frankfurt am Main, Tel: +49(069)154008-873. We are grateful for the comments of participants of the Fourth International Conference on Continuous Optimization, Lisbon, Portugal,  Conference on Current Topics in Mathematical Finance(Poster), Vienna, Austria, Advances in Financial Mathematics, Paris, France, 7th Financial Risks international Forum; Paris, France, CEQURA Conference 2016 on Advances in Financial and Insurance Risk Management,LMU Munich, Germany.\\
}
\end{frontmatter}



\newtheorem{mydef}{Definition}

\section{Introduction}

The recent financial crisis has been the most significant economic event since the great depression in the 1930s and had massive effects on economic growth and on the life of people during and after the financial crisis. It also has raised a wide awareness of key players in a financial system and the requirement of quantitative tools to identify and monitor them. In this context, in a financial system, the determination of the so-called Systemically Important Financial Institutions (SIFI) as strategic players in the market is a main problem.  Institutions (e.g. banks, insurance companies or hedge funds) are systemically important if they are linked with strong transmission channels to others or to institutions who are also systemically key players in the financial system. These institutions are relatively large, vastly interconnected and dominate the system with a large scope of activities.\\

 Although, several methods are proposed, there is no generally accepted way to determine the major institutions in a financial system. The Basel Committee on Banking Supervision (BCBS) (\cite{board2013global}) introduced methods and indicators to identify SIFIs, and reduce the probability of these institutions' failures and smooth the functioning of the financial system. BCBS ranked banks with regards to their size, complexity, interconnectedness and ability to cause substantial disruption to the global financial activity. Financial Stability Board (FSB) and G20 use similar indicators to identify Global Systemically Important Financial Institutions (G-SIFI) and update the list G-SIFI based on new information every year.\\

To reduce the consequences of SIFIs' failure, regulators try to control the scope of activities of these financial institutions to limit the magnitude of their distress or failures to the public sector. On the one hand, recognizing institutions as SIFI has higher costs, which prompted by extra requirements for regulation,  possibility of no flexible approaches by the regulators, and decreasing their competitiveness and economic strength in the domestic and international financial market. On the other hand, it helps regulators to prevent disruptions to financial stability and enhance the country's economic strength. The regulation of SIFI seeks appropriate ways and tools to minimize potential negative effects of the failure of these systemically important and high ranking institutions by restricting their activities and/or re-restructuring their organization.\\

The distress or failure of a SIFI could crystallize systemic risk. A risk of default of a large fraction of the financial system as a whole due to the spread of financial exposures throughout the system is called systemic risk. When the institution is incapable to repay its debtors and fulfill its investor's requirements  it will be in an insolvent position. By reason of interlinked financial exposures among institutions, this distress may spread throughout the whole financial system in a domino fashion and it creates a cascade. This cascade causes a disturbance in the entire system and it may spillover to the larger fraction of the economy (see: \cite{Hellwig2009}).\\

One approach to measure systemic risk is tail measuring which calculates the co-dependence in the tails of institutions’ equity returns without giving consideration to interact with counterparties; \cite{acharya2010measuring};  \cite{brownlees2010volatility} measured systemic risk by measuring the expected shortfall per bank conditional on a distressed financial institution as the systemic capital deficit. \cite{Adrian2008} introduced CoVaR to measure the systemic risk of the financial system conditional on the system being in distress. This method is based on the value at risk (VaR) in the median state of the financial status of an institution. Generally, these types of methods suffer from the limitation of a very few number of extreme values of the magnitude of a financial crisis.\\

Another approach to measuring systemic risk is network analysis. The network approach provides a set of techniques for analyzing the structure of whole entities and a range of theories to explain the mechanisms behind a complex system. \cite{battiston2012debtrank} studied complex networks of the controversial US Federal Reserve Bank (FED) emergency program dataset and introduced  DebtRank metric in a recursive way to determine the impact of the distress of one or some financial institutions through their counterparties network. \cite{Sorama2012} introduced SinkRank metric to predict an influence of disturbance caused by the collapse of a bank in a payment system and identify most affected banks in the system. \cite{cont2012} analyzed the potentials for contagion and systemic risk in an interlinked financial network and introduced a methodology for analyzing an institutional default, metrics for detecting the systemic importance of institutions in a network of interconnected financial institutions. \cite{elliott2014financial} proposed a simple model of cross-holdings to analyze cascades in financial networks; they concluded that diversification and integration of financial institutions have non-monotonic effects on financial contagions. These methods are based on the feedback centrality of the financial network which we will explain later.\\

From this perspective, \cite{allen2000financial} examined the role of the network to the contagion of a small shock and how it spreads to the whole system.  \cite{eisenberg2001systemic} proposed a model of static default cascade which with given a set of banks that are directly interconnected through an interbank network of loans. They found a market-clearing equilibrium, based on interdependencies within a system, by applying fixed point theory. The model built on balance sheet information of banks including nominal (book) values of assets and liabilities, when banks' assets are not sufficient to cover their debt liabilities, banks become insolvent, and finally default. This model assumed insolvent (default) banks cannot repay their counterparties' debt liabilities and assets' value will be reduced.\\

 Another type of cascades begins with an illiquid institution which tries to sell its loans from its debtors, and imposes liquidity stress to its counterparties, triggering a liquidity cascade in the entire system.  \cite{Gairspa20090410} studied a liquidity cascade mechanism to model the credit risk in early 2007 when some big investment banks faced funding liquidity problems since they held partially liquid assets. Such banks started acquiring the liquidity of their counterparties by reducing their interbank lending nearly completely. The static cascades advance through these mechanisms and after a series of deterministic iterations reach "cascade equilibrium" or a steady state of the system.\\

 Network perspective helps to have a system-wide perspective and to design effective mechanisms to prevent the potential breakdown of the financial system. In the aftermath of the great financial crisis, both academic researchers and regulators raise questions regarding how the financial network structure affects the propagation of systemic stress in the economic system, how financial network's linkages are important in determining the consequence of financial interaction, and why several financial institutions suffer concurrently huge losses or default. Particularly, by studying events and  the network structure during the great financial crisis, they want to find out how networks are expected to form when an individual institution has the discretion to select its partners and how the change of key players' characteristics changes roles in the system.\\

Our paper addresses these questions and takes up these challenges by introducing an impact index as a metric of measuring the systemic risk for application to default contagion on (directed) interbank networks. We present a simple model based on the framework of \cite{Gairspa20090410} and the cascade mechanism of \cite{eisenberg2001systemic} model.\\   

 We formulate the problem of finding SIFIs as a fixed point problem, and propose an approximation solution to this problem, and discuss uniqueness and existence of the solution. We then develop a robust network based algorithm to rank financial institutions and to identify the SIFIs in a financial system. We discuss the computational performance of our algorithm and assess in which practical situations this metric is relevant. This algorithm simulates the actual cascade of bank defaults which end up to equilibrium. In our model, nodes of the network hit by individual or aggregate shocks as an exogenous factor. Our model is setup to enable us have a robust estimation of system ranking of financial institutions because it concerns the financial state of an institution as a function of a market condition. Since small events due to a local stress can simply turn into a systemic crisis, the robust identification of SIFIs needs to consider all possible scenarios.\\

Our paper also relates to recent contributions to functions of network to crystallize systemic risk and roles of key players' financial network and systemic risk  by promoting a better understanding of the main factors like the level of interconnection and network position, which lead to an institution to be considered as an SIFI. Past studies have shown that to assess the SIFIs accurately , we need to take the complex and interconnected structure of the financial system into account and focus on the effect of connectivity of the economic network as risk absorbers and/or amplifiers. A speech by Andrew \cite{haldane2013rethinking}, Executive Director for Financial Stability at the Bank of England, pointed out that a higher degree of connectivity of a financial network can absorb shocks, meanwhile it can spread shocks through  the system. One can see there is a disagreement perspective about the role of interconnectivity in the stability of the network. The role of interconnection remains a central open question in the literature. On the one hand, studies by \cite{allen2000financial} and \cite{Freixas2000} showed that a higher interconnected financial system is more stable and by distributing the stress among more members can reduce its side effect on an individual member. Similar to this perspective, \cite{acemoglu2013systemic} studied the link between the financial network structure and the probability of systemic failure. They concluded that in the case of equal distribution of interbank debts and with having a small number of affected institutions, a more densely interconnected system is more stable. On the other hand, some other researches like \cite{vivier2006contagion} argue that the probability of a systemic failure of a financial system might be amplified by a large number of counterparties. A study by \cite{drehmann2013measuring} considered the system-wide importance of an interconnected system and highlighted the difference between systemic impacts of interbank borrowing and lending on a financial system.\\

 This "robust-yet-fragile" characteristic of the financial system leads us to investigate a separate systemic impact of the degrees of out-coming and in-coming of a network member on other members. Thus, we hypothesize that in a financial network, increasing the degree of outcomes and incomes of a network member can increase and decrease the likelihood of that member becoming a SIFI, respectively. We also expect that the systemic aspects of an institution might be explained by means of its network and its dependency to other members of the network. These hypothesizes are tested with the results of our numerical studies.\\

On a broader level, our numerical experiments indicate main conclusions in the application of our framework which is imposing restrictions on the level of interconnections of an institution can significantly reduce the likelihood that the institution is ranked a SIFI.  Nevertheless, raising the degree of out-comes (the number of borrowers) and the degree of in-comes (the number of lenders) of a network member has different impacts on the outcome. The former policy leads to increase by 1.520 percentage points and the later one leads to decrease by 1.685 percentage points that possibility, respectively. These results are robust to different shocks in the network.\\

On the empirical side, we test these hypothesize by using a panel data of the largest banks in each country and study which factors lead to identifying global SIFIs. The empirical outcome is consistent with our theoretical findings which support the hypothesis of the association of the level of the interconnection and the probability of an institution to be classified as an SIFI. The results of the empirical study indicate that raising the likelihood that an institution to be ranked as an SIFI by 1.876 percentage points by increasing one percentage of interconnection levels.\\


\subsection*{Reminder of the paper}
The paper proceeds as follows. Section \ref{SecstatisCascade} 
gives a brief overview of systemic risk and reviews some models of financial network and their general properties. Section \ref{SIFIsection} explains main approaches for identification of SIFIs. Section \ref{Model} explains our model and introduces an impact index as a metric for to identify SIFIs by constructing a fixed problem, and gives an approximation solution to this problem. In section \ref{SectExpFinNet}, we demonstrate how to identify the SIFIs by applying our algorithm in a real interbank network and compare it with an existing method. In section \ref{SecnumbericExp}, we explain the numerical simulation of interbank network and discuss on the efficiency and accuracy of a proposed solution. Section \ref{SecEmprical}  empirically examine the prediction of our theoretical model and investigates factors that lead to identifying Global-SIFIs. Finally, this research is summarized in section \ref{Secsummary}.

\section{The Static Cascade Model}
\label{SecstatisCascade}
In this section, we start with the mechanism of default contagion, explain how the default of one institution may impact other institutions directly, leading to a so-called a default cascade. When the total of the decreasing in the market value of the external assets of a bank exceeds its net worth and the total liabilities of a bank exceed its total assets the institution will be incapable of repaying its debtors and fulfilling its investor's requirements and will be in a default position. This insolvent institution, with the condition not being bailed out by a government action, creates stress for creditor banks' interbank assets. If stress transmitted through the entire system, it will create a channel for default contagion.\\ 

Modern financial system is a concentrated and interconnected system; consequently, it becomes less decomposable and vulnerable to systemically down fall. In this interconnected system, \textit{Contagion Dynamic} can be defined as follows: when a single bank starts to hoard liquidity; it creates challenges for banks that were beforehand borrowing from this bank to meet their own liquidity conditions and have liquidity shortage. As a result, liquidity hoarding can spread across the system due to the structure and connectivity of the interbank network.\\

 Interconnectedness of financial institutions has been highlighted by recent researchers; (\cite{cont2012}; \cite{Gai2010}; \cite{amini2012}) used a network model to study the probability of contagion of financial distressing interbank networks. \cite{hurd2011framework} applied simulation studies based on network models to calculating the expected size of contagion events models. One of the primary sources of systemic risk is the contagion of the economic distress in the financial system. By reason of interlinked financial exposures among institutions, this distress can spread throughout the entire of the financial system in a domino fashion.\\

\cite{hurd2015contagion} described a static cascade model which creates a series of deterministic phases until it reaches to a steady state cascade so-called cascade equilibrium. In the cascade equilibrium, no additional institutions is forced to default.  This static default cascade's model begins with a set of institutions characterized by their balance sheets that are interrelated with the directed network. The balance sheets cover the information of the nominal (book) values of each institution's total assets or/and liabilities. The general assumption is that a default institution may not able to repay the nominal values of its debt, as a consequence the values of interbank assets will be reduced. \cite{eisenberg2001systemic} is a famous model of static cascade which reach a fixed point, equilibrium of the system, as the final phase of the cascade.\\

In the context of microeconomics of banking, "liquidity" is defined as an ability to fulfill requirements while a bank comes due without happening inappropriate losses. Financial institutions have several ways of generating the liquidity; e.g. borrowing from other banks (interbank loans) or the central bank, selling loans and raising capital. Losing liquidity could have happened to a bank if it lost some fraction of its deposits due to the liquidity hoarding of counterparties. Banks, to avoid suffering liquidity problems,  immediately ask for available interbank liquidity. Consequently, the illiquidity spreads through the entire financial system. There are theoretical and empirical evidences that illiquid can lead to being insolvent while banks cannot pay their debts. To guarantee that sufficient liquidity is maintained, it requires that banks monitor the significant individual liquidity risks personally as well as diversifying counterparty risk. A bank wants to be sustainable in the financial system needs to attract significant liquid funds and invests in higher return assets with less market liquidity. Stress testing against extreme scenarios helps them to identify abnormal market liquidity circumstances and avoid sudden liquidity shifts.\\

"Solvency" is defined as an ability of a financial institution to accomplish its long-term expansion. A financial institution becomes insolvent when its current liabilities exceed the current assets, consequently, they are under-capitalized and insufficiency to discharge all debts. Insolvency of a large player in the banking system may amplify the contagion risk to the entire system. For that reason, it is essential to model quantitatively the contagion effect and insolvency of financial institutions and define the basic underlying requirements for controlling solvency risk of the individual institution.\\

The balance sheet of an institution (regulated like a bank or/and no leveraged and no regulated like a hedge fund) consists of the nominal values of assets and liabilities, which the aggregated of securities and contracts. Table \ref{table:balance} schematically characterizes a simplified balance sheet of an institution. Denote by $A^i$ is the nominal value of assets of institution $i$ consists of nominal external assets $(A^E)^i$ (fixed and liquid assets) and interbank asset $(A^{inte})^i$. Similarly, the nominal value of liabilities $L^i$ includes nominal value of deposits $(L^D)^i$ and interbank liability $(L^{inte})^i$. 

\begin {table}[h] 
\captionsetup{font={small,it}}
\begin{center}
\begin{tabular}{c|c}
Assets & Liabilities \\ 
\hline 
Interbank Asset $A^{int}$& Interbank Liability $L^{inte}$ \\ 
External Fixed Asset $A^E$ & Deposits $L^D$\\
$\cdots$ & Capital $Cpa$\\ 
\hline 
\end{tabular} 
\captionsetup{width=.8\textwidth}
\caption {A stylized balance sheet. }
\label{table:balance}
\end{center}
\end {table}
The institution's nominal value equity (capital buffer) can be expressed as:

\begin{eqnarray}
  Cpa^i = (A^{E})^i+ (A^{int})^{i}-(L^{int})^{i}- (L^D)^i.
\end{eqnarray}

When an individual bank is not be capable to meet all its obligations, it may face liquidity shortages, try to use its own reserve or from other sources like interbank markets.

\subsection{Interbank Network}

The recent economic crisis has raised a wide awareness of the financial system need to be studied as a complex network whose nodes are financial institutions, and links are their financial dependencies. Network approach plays a central role in modeling the transmission of the information and determining how it spreads. Network Modeling provides a better view of analysis the dynamical progress of the financial structure and explains how information and financial flow passes through the system. In this perspective, systemic risk can be measured and quantified by the analysis of the dynamics the network.\\

In the aftermath of the great financial crisis, both researchers and regulators raise the question of how the structure of network affects the propagation of systemic stress in an economic system. Network approach can help to create some default scenarios and compare different network structures to understand how a structure of the network affects the propagation of systemic stress in a crisis time.\\

 Network science is mathematically formalized, and some theories and methods are built to study of interactions patterns among nodes. Here, we provide fundamental concepts and definitions that are a basis for the language of financial network analysis. One can find more information about economic network in \cite{Jackson2008}.\\
 
In a general interbank network setting, the matrix $\mathbb{A}$ is defined as an adjacency matrix with the value of interbank exposure if node $i$ and node $j$ are connected, $0$ otherwise:
\begin{equation}
  a_{ij}=\left\{
    \begin{array}{@{} l c @{}}
     1 & \mbox{interbank exposure (debtor i to creditor j)} \\
      0 & \mbox{no interbank relation }
    \end{array}\right.
  \label{eq4}
\end{equation} 

 This adjacency matrix with entries that are not only zero or one represents an interbank position point from a debtor to a creditor. With given the network, we can add up the incoming connections (In-Degree) as a number of creditors and outgoing connections (Out-Degree) as number of debtors. We define these network terminology as follows.\\

\begin{mydef}{\textbf{In-Degree}}\\
In a network, in-degree of one node represents the number of its incoming connections (creditors). The in-degree of the node $i$ is given by:
\begin{equation}
deg^{in}_{i}  ∶= \#\{j \in \mathbb{A}| a_{ji} > 0 \}.
\end{equation}
\end{mydef}

\begin{mydef}{\textbf{Out-Degree}}\\
In this network, the out-degree of one node represents its number of outgoing links (debtors) and it is given by: 
\begin{equation}
deg^{out}_{i}  ∶= \#\{j \in \mathbb{A}| a_{ij} > 0 \}.
\end{equation}
\end{mydef}

The correlation between in- and out-degree of nodes in a financial network indicates how the contagion of economic distress spreads through the entire financial system. The degree distribution of the network, $P_{degree}( deg^{in}, deg^{out})$ is the empirical distribution of the fraction of nodes in a network with in-degree $deg^{in}$ and out-degree $deg^{out}$. \\

The degree of nodes captures how nodes are connected, however, it cannot indicate the nodes position and distance between them. To know which node attracts more attention, we need to look at other properties of a network. These characteristics of the network show how a node is close to other nodes and how easy can reache others. Closeness centrality (\cite{freeman1978centrality}) is a node centrality indice, which is measured by the relative distance of a node to all other nodes. Betweenness centrality index (\cite{freeman1977set}; \cite{anthonisse1971rush}) is another centrality index which indicates who has control over the information or financial flow of the network. We define these centrality measurements as follows.

\begin{mydef}{\textbf{Closeness Centrality}} 
Closeness Centrality (CC) measures how far a node from other nodes with calculating relative distances between nodes. Closeness Centrality (CC) is defined as:
\begin{equation}
CC_i=(n-1)({\sum_j d_{ij}})^{-1},
\end{equation}
where $d_{ij}$ is the shortest distance between nodes $i$ and $j$,in a weighted networks by introducing. \cite{dijkstra1959note}.
\end{mydef}

\begin{mydef}{\textbf{Betweenness Centrality}}
The total number of paths between node $j$ and node $k$ is called $P_{jk}$ and $P_{jk}^i$ is a number of paths pass through node $i$. Betweenness Centrality (BC) of the node $i$ in a network is defined by:
\begin{equation}
BC_i = \binom{n-1}{2}^{-1}\sum_{j,k \neq i} \frac{P_{jk}^i}{P_{jk}}.
\end{equation}
\end{mydef}
Betweenness Centrality measures the intermediary characteristics of a node. All of these measurements capture different aspects of network structure, and none of measurements can dominate others.\\

Due to lack of access to real-world banking network, there only exist a few empirical types of research in this field.  \cite{cont2012} described a methodology to verify potentials for contagion and systemic risk. This study used on an unique dataset of mutual exposures of Brazilian financial institutions to analyze the network structure role in the estimation of systemic risk. The topology of an interbank payment network between commercial banks over the Fedwire Funds Service has been studied by \cite{Soram2007}. \cite{muller2006interbank} assessed the Swiss interbank network and measured systemic risk via mutual interbank exposure and spillover effects of a bank risk on other banks in the system. These empirical studies of structural and topological aspects of real-world networks show that in-degree and out-degree of network's linkages are heavy-tailed distributed and follow power laws distribution with the heterogeneous interconnections structure. The correlation between them has verified the effects of the stress contagion. These studies found certain properties of the network changed considerably immediately after the financial crisis.  Based on these empirical findings, we create a random network model to simulate interbank networks and analysis financial systems.\\


Generally, a financial system can be modeled as a network, that includes a set of nodes $V_{(i \in \{1, \dots, n\})}$ referred to financial institutions and the network is formed as direct or indirect relations. The canonical form a network is an undirected graph, which nodes are connected without direction. This type of network models can represent economic relations or partnership, friendships in social sciences. The second type network is modeled as a direct network in which a node can be connected to another node, in such that the second node is not being connected to the first node. In this directed network topology, node $i$ and the link $\Omega^{ij}$ represent institutions and the interbank exposure, respectively. At time of the default of institution $j$, the institution $i$ faces an exposure cost $\Omega^{ij}$ (directed edge from a debtor to a creditor).\\
 
Total asset value invested by institution $i$ in banking activities is:
\begin{equation*}
A^i= (A^{E})^i + \sum_i \Omega^{ij}.
\end{equation*}

Similarly, total nominal value of liabilities of institution $i$ is:
 
\begin{equation*}
L^i= (L^{D})^i + \sum_j \Omega^{ij}.
\end{equation*}

Matrix $\Pi$ represents the matrix of $n \times n$ interbank exposure matrix (see: Figure \ref{matrixInter}). This matrix contains the information on exposure magnitude a node on its counterparties in the system. The summation of rows and columns gives the information on the interbank liability an assets. 

  \begin{figure}[htbp]
\caption{The Matrix of Interbank Exposure.\label{matrixInter}}
\centering
\[
  \begin{blockarray}{cccccccc}
  &1 &2 & \dots & n &&\\
    \begin{block}{c(cccc)|cc|c@{\hspace*{5pt}}}
       & 0& \Omega^{12} & \dots& \Omega^{1n} & (L^{int})^1 &(L^E)^1 & L^1\\
     &\Omega^{21}& 0 & \dots & \Omega^{2n} & (L^{int})^2& (L^E)^2 & L^2\\
  \Pi=     & \vdots& \vdots & \ddots& \vdots & \vdots& \vdots & \vdots  \\ 
      & \Omega^{n1}&  \Omega^{n2} & \cdots& 0 & (L^{int})^n& (L^E)^n & L^n\\                  
    \end{block}
           \cline{2-5}
    &      (A^{int})^{1}& (A^{int})^{2} & \cdots& (A^{int})^{n} & & & \\
    &    (A^{E})^{1}& (A^{E})^{2} & \cdots& (A^{E})^{n} & & & \\    
            \cline{2-5}
  &   A^1 &A^2 & \dots & A^n &&
  \end{blockarray}
\]
\captionsetup{width=.5\textwidth}
\captionsetup{font={small,it}}
\caption*{The presentation of the matrix of interbank exposure. The summation of elements of columns and rows are interbank assets and liabilities, respectively. }
\end{figure}

Typically, financial networks are large weighted networks. In a weighted network, the adjacency matrix contains information on weights represented the relative nodes strengths. Later, we will discuss about properties of the weighted adjacency matrix and how to use it in our model.

\section{SIFI}
\label{SIFIsection}

As a major effect of the complexity of the contemporary financial network and  lack of an adequate methodology for measuring systemic risk, anticipating the impact of defaults in the financial system becomes a challenging work. For a long time, the size of financial institution's balance sheet has been used to rank the institutions in the system, and large institutions based on their balance sheet size are stated as "Too Big to Fail". However, the recent economic crises has indicated that relatively small institutions  can have the significant impact the financial system.\\

There are two main approaches for identification of SIFIs: model-based approach and indicator-based approach. The first approach is based on main indicators, to comprise broad characteristics of financial institutions. The second one is based on dynamics of financial institutions' relationships. \\

The Basel Committee on Banking Supervision (BCSBS) \cite{board2013global} applied the first approach and developed an indicator-based measurement for assessing the systemic importance of Global Systemic Important Banks (G-SIBs). It is based on five equally weighted indicators: size, complexity ,interconnectedness, substitutability (financial infrastructure) and cross-jurisdictional activity. In April 2009, the Financial Stability Board (FSB) was established to develop strategies for extra regulation of G-SIBs. BCSBS updates the list G-SIB based on information from FSB and publishes the new group of systemic important banks annually. In the US, the Financial Stability Oversight Council (FSOC) (\cite{weistroffer2011identifying}) use slightly different aspects of company's potential including size, lack of alternatives, interconnectedness, leverage, liquidity risk, maturity mismatch and present regulatory inspection. An SIFI may be an investment bank (e.g. Lehman Brother), a hedge fund (e.g. Long Term Capital Management (LTCM)), or an insurer (e.g. AIG). The International Association of Insurance Supervisors (IAIS) developed a similar assessment methodology to identify to Global Systemically Important insurers (G-SIIs) whose can pose the risk to the financial system. BCSBS classified institutions which their activities are limited to a specific region as Domestic (national) Systemically Important Institution (D-SIFI) which have an impact only a specific country's economy. \\

The academic approach to identify SIFI is the model-based approach. The model-based methodology to analyze the systemic importance of financial institutions goes back to theoretical papers by \cite{allen2000financial} and
 \cite{Freixas2000}. Their findings highlighted the significant the role of the structure of the financial network in transmiting countriparty risk.  Model-based approaches (e.g. Systemic Expected Shortfall (SES) suggested by \cite{acharya2010measuring} and COVaR  proposed by \cite{Adrian2008} ), based on market data risk measure  Value at Risk (VaR), Marginal Expected Shortfall (MES) and asset price correlations. \\

Due to lack of sufficient information for all financial institutions, models' instability over business cycles and the complexly of financial system, modeling approach is not robust and cannot capture all the ways that an institution is systemically important. To overcome the traditional methods limitations, some network-based models were developed.  The DebtRank algorithm by \cite{battiston2012debtrank} and SinkRank by \cite{Sorama2012}, as default-cascade approaches, estimated the impact of shocks in a financial network and predicted the magnitude of disturbance triggered by the failure of an institution.  The network model by \cite{eisenberg2001systemic} studied the spread of defaults through an interbank network of nominal obligations  and calculated the clearing payments vector to determine the losses in the network. Based on this model, we propose a metric to identify SIFI with taking into account both common idiosyncratic shocks and contagion through counterparties' exposures.\\

To capture different dimensions of the systemic importance of an institution, we need to measure various aspects of its financial activities. This measurement should be robust in terms of uncertainty in market data quality and exogenous factors, and to be applicable for all types of financial institutions. Indicator-based measurement is a simple method to rank institutions and is robust with respect to some uncertainties of the modeling approach. Nonetheless, this approach is static and cannot reflect dynamics of the financial system.\\

Basel III agreement introduced a set of tools and methodology to regulate the SIFIs, to reduce the probability of failure and limit SIFIs failures consequences. As a result, it can enhance the financial stability and mitigate the cost of governments' interventions. These tools can be classified in following categories: (i) the reduction of the SIFIs size: institutions' size represents the potential of posing risk to the entire system partially, (ii) re-organizing the SIFIs' structure can reduce the systemic importance of an institution, (iii) limiting of the institution activity can reduce the contagion of stress through the system, (iv) an additional capital buffer can reduce the probability of SIFIs' failure and (v) improving store liquidation of SIFIs may help in the event of failure. The scope of these regulations is limited to available information on institutions' characteristics and their activities.\\


\section{The Model}
\label{Model}
%
%

In this section, we explain our model and introduce an index to quantify financial relations among institutions, regarding the impact of the distress of institutions to their counter-parties across the entire system in a recursive way. In this approach, we determine systemically important institutions by analyzing the dynamics of their dependencies. We use a directed network where the vertices represent financial institutions, and the edges represent financial relations.\\

Our model studies network effects of liquidity hoarding procedure. Mainly, it shows how liquidity shortages can spread through the financial system via interbank linkages. The model examines the financial system structure, such as joint distribution of lending and borrowing links and connectivity degrees to understand how shocks spread.\\

 The system at time $t=0$ is in a normal state and the network is formed. At time $t=1$, the system is hit by a shock or several shocks which have impacts network. The stressed nodes start to transmit shocks to their counterparties and cascade begins. Basically, we assume there is no channel for fire sale and changing of all external cash flows, assets throughout the cascade are not considered. At time $t=2$, the cascade continues until it reaches equilibrium (See Figure: \ref{figModel}).\\
 
  \begin{figure}[htbp]
\centering
\includegraphics[scale=.45]{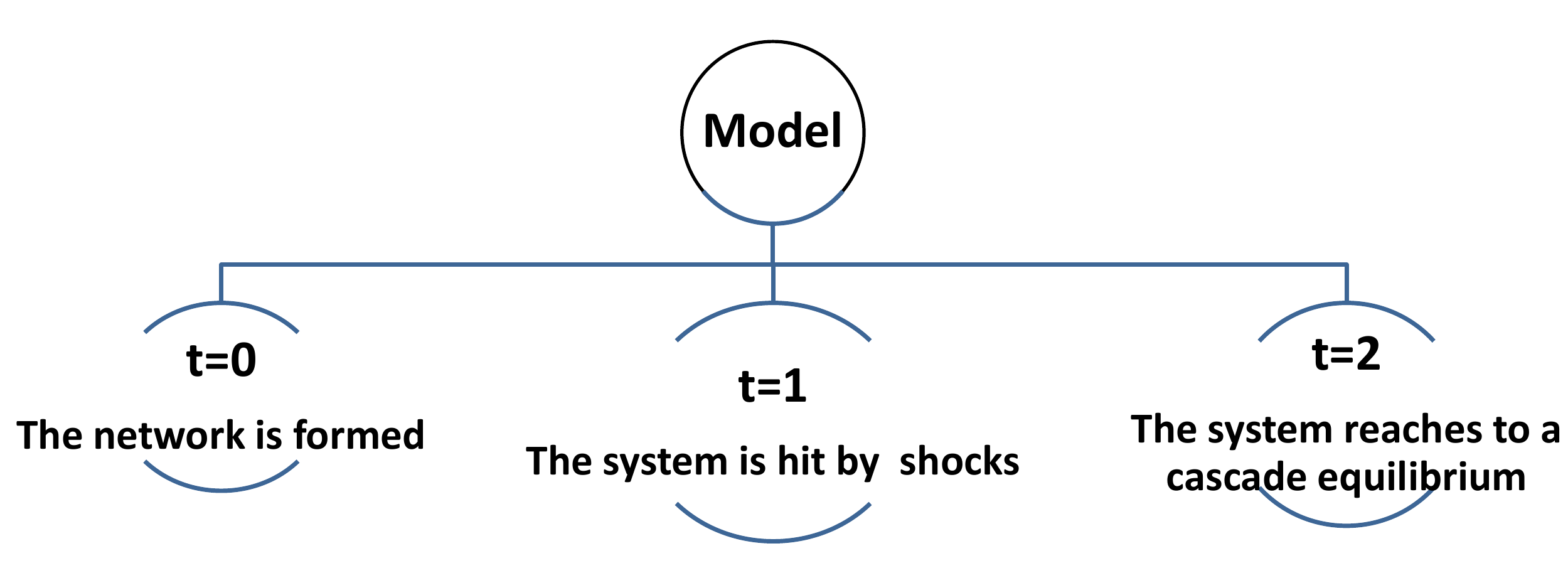} 
\caption{The model \label{figModel}}
\end{figure}

Our network setting (economy) has a simple structure: nodes contain information on balance sheets of institutions and weighted direct edges represent the amount of exposure of a node on its counterparties. In this economy, we use solvency condition index of an institution, expressed in \cite{Gairspa20090410}, when
the current institution's liabilities exceed its current assets. A bank becomes insolvent when its current liabilities exceed its current assets.
\cite{Gairspa20090410} defines the solvency condition as follows:
\begin{eqnarray}
\label{giakapadia}
 (1-\phi)(A^{int})^i + (A^E)^i -(L^{int})^i - (L^D)^i>0,
\end{eqnarray}
where $\phi$ is the portion of banks' obligation to a defaulted bank. We can determine the solvency condition of bank ($i$) with having an assumption which is losing of interbank asset of an instituation at time of default (\cite{Gairspa20090410}).

From \cite{Gairspa20090410} solvency condition, equation \eqref{giakapadia}, we have:
\begin{eqnarray}
\label{orginalSolvency}
 \frac{(1-\phi)(A^{int})^i + (A^E)^i -(L^D)^i}{(L^{int})^i}>1.
\end{eqnarray}
This condition shows the ability of an institution over a cash flow time horizon to meet its obligation when debts are collected by its debtors. With using network centrality measurements, we define solvency index as following:

\begin{mydef}{\textbf{Solvency Index}}\\
 Let $ \varepsilon^{sh}$ be the aggregate or individual deposits' shocks. From equation \eqref{orginalSolvency} we have:
\begin{eqnarray}
\chi_i &=&  \frac{(1-\phi)(A^{int})^i + (A^E)^i -(L^D)^i+ \varepsilon^{sh}}{(L^{int})^i}\\
      &\approx& \frac{(1-\phi)deg^{in}_{i}(\bar{\Omega})^i + (A^E)^i -(L^D)^i+\varepsilon^{sh}}{deg^{out}_{i}(\bar{\Omega})^i},
\end{eqnarray}
where $(\bar{\Omega})^i$ is the average of interbank exposures of node $i$ and $deg^{in}_{i}$ and $deg^{out}_{i}$ are in-degree and out-degree of node $i$.\\
Solvency Index of node $i$ ($Solv_i$) is defined as  institution ($i$)'s inability to satisfy the solvency condition and then contagion starts to spread. This index cab be expressed by:

\begin{eqnarray}
Solv_i=\min (\chi_i, 1).
\end{eqnarray}
\end{mydef}

In a counterparty network, the vulnerability weight of a node can be defined as a fraction of capital which is decreased by the impact of the default of its counterparty. With given the capital buffer of the institution $i$ $Cpa_i$, which is used against financial risk, we define the vulnerability weight index as the following:

\begin{mydef}{\textbf{Vulnerability Weight}}\\
Vulnerability weight of the institution  $i$ is expressed as its ability to mitigate the interbank activities tensions which make the financial system susceptible to counterparty risk. It can be quantified as the proportion of the capital of an institution used against interbank exposure:
\begin{equation}
\mathbf{w}_{ij}=\min(\frac{\Omega^{ij}}{Cpa^i},1).
\end{equation}
\end{mydef}

The economic value of the impact of one institution on others is calculated by multiplying the impact by the relative economic value of its counterparties. We define \textit{Relative Impact} as a fraction of solvency index one institution on accumulative solvency indexes of its counterparties.\\

The impact value of an individual institution on its counterparties can be defined as:
\begin{equation}
\mathbf{\tilde{p}}_i=\sum_j \mathbf{w}_{ij} \mathbf{r}_i,
\end{equation}
 where $\mathbf{r}$ is the solvency ratio which represents the relative solvency of an institution to its counterparties:
\begin{equation}
\label{counterpartie_impact}
\mathbf{r}_i=\frac {Solv_i}{\sum_j Solv_i} \qquad  \quad \{j|a_{ij}>0\}.
\end{equation}

In addition, an institution's counterparties have an impact, we add the second term in equation \eqref{counterpartie_impact} for an indirect impact via neighbors. We define the corresponding \textit{Impact Index} as measuring of the expected loss, which is generated by the failure of a group of institutions as follows.
\begin{mydef}{\textbf{Impact Index}}\\

Similar to \cite{battiston2012debtrank}, we define the impact index of institution $i$ as follows:\\

\begin{equation}
\label{CIeq}
 \mathbf{p}_i= \underbrace{\alpha  \sum_j \mathbf{w}^{BC_i}_{ij} \mathbf{r}_i}_{\text{expected loss of institution's failure (direct effects)}} + \underbrace{\beta \sum_j \mathbf{w}^{CC_i}_{ij}  \mathbf{p}_j}_{\text{expected loss by counter-parties' failure (indirect effects)}}\\
( \beta = 1- \alpha),
\end{equation}
where $\alpha< 1$ is the probability of imposing shock by a node and $\beta = 1 - \alpha$ is the probability of reviving shocks by counterparties , and $CC_i$ and $BC_i$ are the closeness and Betweenness centrality of the node (institution) $i$.
\end{mydef}

 The first term is measuring the proportional of the expected loss, generated by the failure of an institution (direct effects), and the second proportion can be coming from the failure of its counterparties (indirect effects). As above mentioned, in a network the closeness centrality of a node indicates the relative distance of the node to other nodes and another centrality measurement which is called betweenness centrality indicates who has control over flow between others. In other words, these indexes can show the scope of impacts of an institution's distress or failure on the system. With the higher value of  betweenness centrality of  a node shows the vulnerability of a node to control the transmission of counterparties' stress.  Similarly, the higher value of the closeness of a node can intensify the impact of a node on its counterparties.

\subsection{Problem Formulation}

The vector $\mathbf{p}$ contains information on impact indexes of institutions in the network, we can rewrite equation \eqref{CIeq} in a matrix form:
\begin{equation}
\label{matEq}
  \mathbf{p}= \mathbf{M} \mathbf{p},
\end{equation}
where let matrix $M$ defined on $R^{n\times n}$ as a non-negative (positive) matrix from  equation \eqref{matEq}:
\begin{equation}
\label{Mmatrixeq}
\mathbf{M}= \alpha \mathbf{W}^{CC} + \beta \mathbf{W}^{BC}  \mathbf{R}, \quad \quad \quad ( \beta = 1- \alpha),
\end{equation}
where the matrix $\mathbf{R}$ is a diagonal matrix with entries on the main diagonal are the vector $r$.  We can decompose a perturbed matrix $\mathbf{M}$ into two matrices: nominal matrix $\mathbf{P}$ which contains information on vulnerability weighted of impact indexes of institutions and a perturbation matrix $\mathbf{\Xi}$ which includes information on the aggregate or individual deposits' shocks $ \varepsilon^{sh}$.\\

\begin{equation}
\mathbf{M}= \mathbf{P} +  \mathbf{\Xi}
\end{equation}

Let $\mathcal{P}$ be a set of perturbed matrix and $\mathbf{M} \subset \mathcal{P} $.
We are looking for a robust solution  $\bar{\mathbf{p}}$ of the dominant eigenvector of matrix $\mathbf{M}$.

Indeed, finding central nodes which pose counterparty risk to the system can be used to identify systemically important financial institutions in the financial system. We can formulate this problem as a fixed point problem and find a robust Eigenvalue solution. In the following proposition we define a robust Eigenvalue solution to this fixed point problem. \\
\newcommand{\Argmin}{\operatornamewithlimits{Argmin}}.
\newtheorem{myprop}{Proposition}
  \begin{myprop}{\textbf{Robust Eigenvalue solution}}\\
 The vector $\bar{\mathbf{p}}$ is defined as
\begin{equation}
\label{mainprob1}
\bar{\mathbf{p}}= \mathbf{M}  \bar{\mathbf{p}}
\end{equation}
is a robust solution of the fixed point problem on matrix $\mathbf{M}$ if
\begin{equation}
\label{mainprob2}
\bar{\mathbf{p}} \in \Argmin_{\mathbf{p} \in \mathbf{\Sigma}} (\max_{ \varepsilon^{sh} \in \mathbf{\Xi}}||(\mathbf{P}+\mathbf{\Xi})  \mathbf{p} - \mathbf{p} ||_2  \quad  | \quad ||\mathbf{\xi}||_1 \leq \epsilon),
\end{equation}
where $\mathbf{\Sigma} :={ \{u^2 \in R^{n \times n}| \sum_i u_i=1\}}$, $ ||\mathbf{\xi}||= \sum_{ij}|\mathbf{\xi}_{ij}|$ and $\epsilon >0$.\\
\end{myprop}
 
\begin{proof}
 Let $\mathbf{\Sigma}$ be a set of eigenvectors, we rewrite the equation \eqref{mainprob1} as a convex optimization problem with measuring the goodness of the solution in norm $||.||_2$:

\begin{equation}
\label{mainprob2}
\min_{\mathbf{p} \subset \mathbf{\Sigma}}||\mathbf{M}  \mathbf{p} - \mathbf{p} ||_2.
\end{equation}

Matrix $\mathbf{M}= \mathbf{P} +  \mathbf{\Xi}$ be a perturbed matrix and a subset of $\mathcal{P}$ and matrix $\mathbf{\Xi}$ is a perturbation and includes information on  $ \varepsilon^{sh}$. In other words this information is representation of uncertainty about the aggregate or individual deposits. \\
The fixed point problem \eqref{mainprob1} can be reduced to a convex optimization problem \eqref{mainprob2} and it can be reformulated as finding a dominant eigenvector of a positive matrix.\\

 With fixing total amount of uncertainty by  $||\mathbf{\xi}||_1= \sum_{ij}\mathbf{\xi}_{ij} < \epsilon$, a robust solution of the fixed point problem \eqref{mainprob1} can be expressed as a solution of the following min-max problem:

\begin{equation}
\label{mainprob25}
\min_{\mathbf{p} \subset \mathbf{\Sigma}}(\max_{\mathbf{M} \subset \mathcal{P}}||\mathbf{M}  \mathbf{p} - \mathbf{p} ||_2). 
\end{equation}
%
%
\end{proof}
[\cite{Juditsky2012}, Proposition 2] proved in the case norm 2, vector $\bar{\mathbf{p}}$ is an unique solution due to the strict convexity on $R^{n \times n}$. In fact $\bar{\mathbf{p}}$ coincides with $\mathbf{p}$ if $\epsilon$ is small enough and all eigenvalues lay inside the $1-\lambda_i >0; i = 2; \dots; n$.\\


\subsection{Solution (Successive Approximation Method)}

The primary fixed point problem \eqref{mainprob1} is formulated as the problem of finding a robust score vector $\bar{\mathbf{p}}$. Matrix $\mathbf{M}$ is a weighted adjacency matrix of a graph with $n$ connected nodes with outgoing and incoming links. One can apply standard eigenvector centrality and PageRank methods (\cite{page1998}) to this problem, however, there exists several issues related using these methods. Firstly, matrix $\mathbf{M}$ is neither column-stochastic nor row-stochastic i.e. it is not a square matrix of non-negative with each row or each column summing to one. In this case, vector $\bar{\mathbf{p}}$ is not unique: typically in the real-world financial network, there are some disconnected subgraphs as well as cycles and solution for cyclic matrix it may converge slowly. Secondly, financial networks are large weighted networks, which are often perceived as being harder to analyze compare to their unweighted counterparts. Thus the eigenvector centrality and standard PageRank cannot detect cascade in this system.\\

Since matrix $\mathbf{M}$ is an irregular matrix i.e. there is no robust dominant vector $\bar{\mathbf{p}}$ in an explicit form, therefore, we solve the structured convex optimization problem \eqref{mainprob2} approximately. For a medium-size problem in a medium size financial network, we can apply interior-point methods as a continuous nonlinear optimization methodology. In a large-scale financial network, this problem can be solved with a method called mirror-descent family ( see: \cite{Polyak2009}\cite{Lan2011}. For a huge counterparty’s network, if the adjacency matrix is sparse enough, the method that is explained in \cite{Nesterov2012} can be applied. \cite{Juditsky2012} found a numerical approximation solution to the similar eigenvector problem, and it works based on the power method. This method can be used in the case adjacency matrix has linearly independent eigenvectors and the eigenvalues can be ordered in magnitude a $|\lambda_1|>|\lambda_2|>\dots>|\lambda_n|$.\\

Similar to \cite{eisenberg2001systemic},  vector $\mathbf{p}^k$ denotes amount of systemic impact of institutions in the system at the end of the $k^{th}$ step of the cascade.  \cite{hurd2015contagion} explained "cascade equilibrium" of the system while there is not default occurred in the system. These series of finite steps will be monotonically continued until the system converges to a fixed point in a cascade equilibrium. We use the power method as a standard algorithm to compute the dominate eigenvector of a eigenvalue problem. This method can be described briefly as follows: in the first step of initial values of  $\mathbf{p}^0$ should be chosen. In the second step, the dominate term  $\sigma_{k+1}$ of  $\mathbf{M}  \mathbf{p}^k$ should be determined. Finally, we compute the next dominate eigenvector: $\mathbf{p}^{k+1} = \frac{1}{\sigma_{k+1}} \mathbf{M} \mathbf{p}^k$. After several iterations the result of the algorithm converges to largest eigenvalue which is called the dominant eigenvalue of the matrix $\mathbf{M}$. One of the main advantage of this method is that it can be done with any norm without normalizing the matrix.\\

With using the basic idea of the power method, we present a successive approximation method for computing vector $\bar{\mathbf{p}}$ as ranks of systemically important financial institutions, in the following proposition: \\

\begin{myprop}{\textbf{Impact Index}}\\

Impact Index for (directed) interbank networks can be determined as a solution of the fixed point iteration:

\begin{equation}
  \mathbf{p}^{k+1}= \frac{k}{k+1} \mathbf{W}^{CC}  \mathbf{p}^{k} + \frac{1}{k+1} \mathbf{W}^{BC}  \mathbf{r},
\end{equation}
where $k$ is the number of iterations needed in order the system converges to a fixed point monotonically.
\end{myprop}

\begin{proof}
Suppose that vector $\mathbf{p}^k$ denotes information of systemic impact of institutions at the end of the $k^{th}$ step of the cascade.\\
 We can write the vector $\mathbf{p}^k$ as a linear combination of vectors from previous steps of the cascade:
\begin{eqnarray*}
\mathbf{p}^k &=&  \frac{\mathbf{p}^1+ \mathbf{p}^2 + \dots +  \mathbf{p}^{k-1}}{k}\\
        &=& \frac{\mathbf{p}^1+\mathbf{M}^1  \mathbf{p}^1 + \dots + \mathbf{M}^{k-1}  \mathbf{p}^1}{k}.
\end{eqnarray*}

We apply the power method with averaging to minimize problem \eqref{mainprob2}:
\begin{equation*}
\min{|| {\mathbf{M} \mathbf{p}}-\mathbf{p}||_2 = \min\frac{|| \mathbf{M}^k \mathbf{p}_1- \mathbf{p}_1||_2 }{k}\leq \frac{Constant}{k}},
\end{equation*}
By choosing $\mathbf{p}^1= \mathbf{W}^{BC} r$ and from equation \eqref{Mmatrixeq}, we will get this recursive rule format:
\begin{equation*}
\mathbf{p}^{k+1}=\frac{k}{k+1}\mathbf{W}^{CC} \mathbf{p}^k+ \frac{1}{k+1} \mathbf{W}^{BC} r.
\end{equation*}
With the following equation \eqref{Mmatrixeq} we can define matrix $\mathbf{M}^k$ as :
\begin{equation*}
\mathbf{M}^k=\frac{k}{k+1} \mathbf{W}^{CC} + \frac{1}{k+1} \mathbf{W}^{BC}R .
\end{equation*} 

At each step of the cascade, the impact index vector is computed as follows:

\begin{equation}
\mathbf{p}^{k+1}=\mathbf{M}^k \mathbf{p}^k
\end{equation}

That means $\alpha=\frac{k}{k+1}$ and $ \beta=1-\frac{k}{k+1}$ are weighting coefficients of matrix and representing the probabilities of direct and indirect shocks of a node.
\end{proof}
To construct a solution, we use a successive approximation to find a fixed point, and state and prove the existence and uniqueness of the equilibrium by means of the contraction mapping theorem.\\

Let $\mathbf{X}$ be an economy (financial system) defined as a set of impact indexes of institutions $\{1,\cdots,n\}$ denote by $\mathbf{X}_1,\mathbf{X}_2, \cdots, \mathbf{X}_n$ \footnote{In this section, for the sake of simplicity we change the notation.}. Let $\Phi : \mathbf{X} \to \textbf{X}$ be a continuous transformation function that maps $\mathbf{X}$ onto itself. A fixed point of $\Phi$ is an element $\mathbf{X}_i \in \mathbf{X} $ for which $\Phi(\mathbf{X}) = \mathbf{X}$. We use the notation $\Phi(\mathbf{X}_i) = \mathbf{X}_i$ in place of $\mathbf{X}_i=M^T \mathbf{X}_i$.\\

To compute the successive approximations and prove the existence of a fixed point, we use an iterative scheme: $\mathbf{X}_{n+1}= \Phi(\mathbf{X}_n)_{\{n=1,2,\cdots\}}$. 
With using the Brouwer fixed-point theorem and defining sufficient conditions on $\Phi$ and $\mathbf{M}$, we prove the existence and uniqueness of the systemic impact index in a banking network.\\

\begin{myprop}{\textbf{Existence and Uniqueness of the Systemic Impact Index}}\\
\label{existenceUniq}
In a financial system consists of $N$ firms with interbank activities being represented by matrix $\mathbf{M}$ and by a continuous transfer function $\Phi$,\\
(a) there exists an index associated with the ranking of systemically important firms in the network;\\
(b) transfer function $\Phi$ has a unique fixed point in matrix $\mathbf{M}$.
\end{myprop}

\begin{proof}
\textbf{Part (a):} $\Phi : \mathbf{X} \to \mathbf{X}$ is defined as a continuous, transformation function that maps $X$ onto itself. Based on Brouwer's fixed-point theorem, $\Phi$ is a continuous mapping function. Therefore, there exists a fixed point $\mathbf{X}_i \in \mathbf{X} $.\\

Vector $\mathbf{X}$ is an economy that is a set of firms corresponding to an interbank equilibrium of a financial network. Note that any interbank equilibrium in this model is a fixed point of $\Phi$.\\

\textbf{Part (b):} if there exist two fixed points $\mathbf{X}_1$ and $\mathbf{X}_2$,
\begin{equation}
   d(\mathbf{X}_1,\mathbf{X}_2) = d(\Phi (\mathbf{X}_1), \Phi (\mathbf{X}_2))   < \theta d(\mathbf{X}_1,\mathbf{X}_2),
\end{equation}
\newcommand{\argmin}{\operatornamewithlimits{argmin}}
where $d(\mathbf{X}_1,\mathbf{X}_2) = || \mathbf{X}_1-\mathbf{X}_2||_2$ with condition $0 \leq \theta < 1$. We should prove that $d(\mathbf{X}_1,\mathbf{X}_2)  \to 0$.\\

According to the Perron--Frobenius theorem, proved by \cite{perron1907theorie} and \cite{frobenius1912matrizen}, for matrix $\mathbf{M}$ (a real square matrix with non-negative entries and an irreducible matrix since from a connected graph), there exists a dominant eigenvector $\bar{\mathbf{X}}$ with strictly positive components, which can be expressed as
\begin{equation*}
\hat{\mathbf{X}} = \argmin_{\mathbf{X} \in \sum} \lbrace || {\bar{\mathbf{X}}- \mathbf{M} \bar{\mathbf{X}}}||_2 + \epsilon\rbrace.
\end{equation*}
Similarly, we have $ \mathbf{M}\mathbf{X} = \mathbf{X}- \epsilon$, and $\hat{\mathbf{X}}$ coincides with $\bar{x}$ if $\epsilon$ is small enough. Then, we can obtain a linear convergence in the fixed-point iterations.\\

Given the mapping function $\Phi$ on $X \to X$,
\begin{eqnarray*}
 d(\mathbf{X}_1,\mathbf{X}_2)   &=& d(\Phi (\mathbf{X}_1), \Phi (\mathbf{X}_2)) \\
             &=& d (\mathbf{M}\mathbf{X}_1 , \mathbf{M}\mathbf{X}_2) \\
             &=& d (\mathbf{X}_1-\epsilon _1,\mathbf{X}_2-\epsilon _2) \\
             &<& d (\mathbf{X}_1,\mathbf{X}_2)-d(\epsilon _1,\epsilon _2) \\
             &\leq& \theta d (\mathbf{X}_1,\mathbf{X}_2),
\end{eqnarray*}
we can write the above equation in a general form
\begin{eqnarray*}
 d(\mathbf{X}_n,\mathbf{X}_{n+1})   &=& d(\Phi (\mathbf{X}_n), \Phi (\mathbf{X}_{n+1})) \\
             &\leq& \theta d (\mathbf{X}_n,\mathbf{X}_{n+1}).
\end{eqnarray*}
The recurrence form gives us
   \begin{eqnarray*}
    d(\mathbf{X}_n,\mathbf{X}_{n+1})   &\leq& \theta^n d (\mathbf{X}_0,\mathbf{X}_{1})  \\(n = 1,2, ... ),
   \end{eqnarray*}
which implies $ d(\mathbf{X}_p,\mathbf{X}_q) \to 0$ for any arbitrary $p,q$ in $n = 1,2, \cdots$. Thus, the interbank equilibria of the model is basically unique.
\end{proof}

\subsection{Algorithm}
After we model the finding SIFIs in an interbank network as a fixed point problem formulate it as a convex optimization problem, we explain the approximation solution. Now we specify an algorithm to describe the mathematical process of solving this problem with approximated solution.\\

In the initial step of this proposed algorithm, we simulate a financial network with a complete or scaled free random networks, and then we calculate adjacency matrix corresponding to each network. We define the impact of one node on its indirect successors in terms of a recursive equation: $\mathbf{p}^{k+1}= \mathbf{M}^k  \mathbf{p}^k$, where k is the iteration number of solution procedure. With given a weighted adjacency matrix of the interbank network, the first step of this procedure is specifying the initial value of impact index.\\

Following algorithm explains this procedure:

\subsubsection*{ Algorithm 1}
\label{algo1}
\textbf{Step 1}:  begin: $\mathbf{p}^1$ as initial value for impact index \\

\begin{center}
$ \mathbf{p}^1= \mathbf{W}^{BC}r$.\\

\end{center}
 \textbf{Step 2}: k-th iteration: $\mathbf{p}^k=\mathbf{M}^k   \mathbf{p}^{k+1}$,\\
where\\
\begin{center}
$ \mathbf{M}^k=\frac{k}{k+1}  \mathbf{W}^{CC} + \frac{1}{k+1}  \mathbf{W}^{BC}  r$.\\

\end{center}
\textbf{Step 3}: Stop condition:  $|\mathbf{p}^k - \mathbf{p}^{k+1}|< Tolerance$, choosing tolerance is dependent upon the size of the problem. \\

In order to get a robust ranking index, we repeat above algorithm several times and each time the system hits by a shock. In the section, we discuss the underline distribution of Impact Index and show it converge in average to robust ranking index. 

\section{Statistical Inference on Impact Index}

In this section, we study properties of an underlying distribution of the impact index  and we model the probability of an institution considered as a SIFI in a financial network. In simulation section we estimate this probability condition on SIFIs' characteristics.\\

Suppose that impact indexes of institutions $\{1,\cdots,n\}$  $\mathbf{X}_1,\mathbf{X}_2, \cdots, \mathbf{X}_n$ are independent and identically distributed (i.i.d.) sample from a distribution function $\mathcal{F}(x) = \mathbb{P}(\mathbf{X} \leq x)$.
The function $\mathcal{F}(x)$ is an unknown function, and there is no hypothesis about its' parameters' value.\\

We define the empirical (cumulative) distribution, with mass $\frac{1}{n}$ for each data point as follows:
 
\begin{equation*}
\tilde{\mathcal{F}}_n (x)= \frac{1}{n}\sum_{i=1}^n \mathbf{1} \{\mathbf{X}_i \leq x \}, 
\end{equation*}
where the indicator function is defined by

\begin{equation*}
\mathbf{1} \{\mathbf{X}_i \leq x \} =\begin{cases}
    1, & \text{if $\mathbf{X}_i \leq x$},\\
    0, & \text{otherwise}.
  \end{cases}
\end{equation*}

According to the fundamental theorem of statistics (Glivenko-Cantelli Theorem \cite{ tucker1959generalization}) we can determine the asymptotic behavior of the empirical distribution function $\tilde{\mathcal{F}}_n (x)$ which it converges to $\mathcal{F}(x)$ almost surely by the strong law of large numbers as $n \to \infty$:

\begin{equation*}
\sup_{x \in \mathbb{R}} |\tilde{\mathcal{F}}_n (x) - \mathcal{F}(x) | \rightarrow 0 \quad \text{almost surely},
\end{equation*}

To calculate a confidence for distribution function $\mathcal{F}(x)$, find a confidence region $C(x)$ such that 

\begin{eqnarray}
 \mathbb{P}(\mathcal{F}(x) \in C(x)) \geq 1- \alpha', \quad \quad \forall x \in \mathbb{R},  
\end{eqnarray}
we use a result of Dvoretzky-Kiefer-Wolfowitz (DKW ) inequality \cite{massart1990tight}:
\begin{eqnarray}
 \mathbb{P}\{\sup_x|\mathcal{F}(x)- \tilde{\mathcal{F}}(x)| > \epsilon \} \leq 2 \exp(-2n\epsilon^2) = \alpha'.
\end{eqnarray}
Then, a confidence band for function $\mathcal{F}(x)$ for all $n$ is:

\begin{eqnarray}
\mathbb{P}(L(x) \leq F(x) \leq U(x)) \geq 1-\alpha', \quad \quad \forall x \in \mathbb{R} 
\end{eqnarray} 
where  
\begin{eqnarray}
L(x) = \max \{\tilde{\mathcal{F}}_n(x) - \sqrt{\frac{1}{2n}log \frac{2}{\alpha'}},0  \},
\end{eqnarray}
\begin{eqnarray}
U(x) = \min \{\tilde{\mathcal{F}}_n(x) + \sqrt{\frac{1}{2n}log \frac{2}{\alpha'}},1  \}.
\end{eqnarray}
\subsection{Threshold of Classification}
In order to identify SIFI, we need to sort impact indexes of institutions  and define a threshold for each class of SIFI. This threshold can  be defined by the quantile value of distribution function $\mathcal{F}(x)$.\\

 Let $\rho : \{1,\cdots,n\} \longrightarrow \{1,\cdots,n\}$ be a operator such that $\mathbf{X}_{\rho(i)} \leq \mathbf{X}_{\rho(j)}$ if $i < j$.\\
 
After ordering impact indexes, we define the order statistics ($\mathbf{X}_{\rho(i)}=\mathbf{X}_{(i)}$) as follows:

\begin{equation*}
\mathbf{X}_{(1)} \leq\mathbf{X}_{(2)} \leq \cdots \leq  \mathbf{X}_{(n)}.
\end{equation*}
 
Let $u_p$ be the quantile which divides the CDF $ \mathcal{F}(x)$ into two parts, and it is defined as:

\begin{equation}
\label{quantileEq}
u_p = \mathcal{F}^{-1}(p)= \inf\{x: \mathcal{F}(x) \geq p \}, 
\end{equation}
where the function $\mathcal{F}^{-1}:(0,1) \rightarrow \mathbb{R}$ is a quantile function. \\

This function is left-continuous function and (\cite{van2000asymptotic} Lemma 21.1) stated that for any series of CDFs, $\tilde{\mathcal{F}}^{-1}_n (x)$ converges to $\mathcal{F}^{-1}(x)$ almost surely  if $\tilde{\mathcal{F}}_n (x)$ converges almost surely to $\mathcal{F}(x)$. We use this result and define $\tilde{u}_{pn}$ as an estimator of quantile $u_p$:\\

\begin{equation}
\tilde{u}_{pn} = \tilde{\mathcal{F}}_n^{-1}(p)= \inf\{x: \tilde{\mathcal{F}}_n(x) \geq p \}, 
\end{equation}

We identify  SIFIs who their impact index less or equal to the quantile $u_p$. If at least $r$ number of institutions with impact index less or equal to $u_p$, with the probability of an institution classified as a SIFI $\mathbb{P}(\mathbf{X}_{i}\leq u_p)  = \mathcal{F}(u_p) $ , the probability of $M= \sum_{i=1}^n \mathbf{1} \{\mathbf{X}_i \leq u_p \}$ institutions are ranked as SIFI is:

\begin{eqnarray}
\label{pronsifi}
     \mathbb{P}(M \geq r) &=& \mathbb{P}(\mathbf{X}_{(r)} \leq u_p ) \\      
                 &=& \mathbb{P}(\mathbf{X}_{(r)}\leq  \tilde{u}_{pn} )  \quad \quad \mbox{ for all } n  \\                       &=& \sum_{i=r}^n    \binom {n} {i} \tilde{\mathcal{F}}_n(u_{pn})^i (1-\tilde{\mathcal{F}}_n(\tilde{u}_{pn}))^{n-i}.
\end{eqnarray}

To have a robust estimation of the threshold, we construct a  100(1-$\alpha'$)\% confidence interval for quantile $u_p$ with any $k,l \in \mathbb{N}$ 
  \begin{eqnarray}
  \label{probsifi}
 \mathbb{P}(\mathbf{X}_{(k)} \leq u_p \leq \mathbf{X}_{(l)}) &=& 1- \alpha'.
 \end{eqnarray} 

Accourding to results of [\cite{van2000asymptotic} chapter 21], and by using the strong law of large numbers as $n \to \infty$:

\begin{eqnarray*}
\lim_{n \to \infty} \tilde{u}_{pn}&=& \tilde{\mathcal{F}}^{-1}_n (p)\\
 &=&  \mathcal{F}^{-1}(p)  \quad \text{almost surely}\\
 &=& u_p.
\end{eqnarray*}

By applying the probability integral transformation we define random variables ${\mathcal{U}_i}_{\{1,\cdots ,n\}} = {\tilde{\mathcal{F}}_n(X_i)}_{\{1,\cdots ,n\}}$ as samples of uniform distribution, we have:
 
 \begin{eqnarray*}
 \mathbb{P}(\mathbf{X}_{(k)} \leq u_p \leq \mathbf{X}_{(l)}) &\approx& \mathbb{P}(\mathbf{X}_{(k)} \leq \tilde{\mathcal{F}}_n^{-1}(p) \leq \mathbf{X}_{(l)})\\
                  &=& \mathbb{P}(\tilde{\mathcal{F}}_n(\mathbf{X}_{(k)}) \leq p \leq \tilde{\mathcal{F}}_n(\mathbf{X}_{(l)}))\\
                 &=& \mathbb{P}(\mathcal{U}_{(\mathbf{X}_{(k)})} \leq p \leq \mathcal{U}_{(\mathbf{X}_{(l)}))}
\end{eqnarray*}   

To construct two sided confidence band we choose $k < l$ and divide mass $\alpha'/2$. We determine the upper confidence limit of quantile $u_p$ as follows:
 \begin{eqnarray*}
l =  \inf \{ i \in \{ 1,\cdots ,n-1\}:  \mathbb{P}(u_p \leq \mathbf{X}_{(i)}) \leq \alpha'/2 \}.
\end{eqnarray*}  
Similarly, the lower confidence limit of quantile $u_p$ is:
 \begin{eqnarray*}
k =  \sup \{ i \in \{ 1,\cdots ,n-1\}:  \mathbb{P}(u_p \leq \mathbf{X}_{(i)}) \geq 1 - \alpha'/2 \}.
\end{eqnarray*}  
\subsection{Asymptotic Distribution}
We can use the result of equation: \eqref{pronsifi} to obtain asymptotic distribution of impact index for larger size networks.\\

Let $\tilde{u}_{pn}$ be the estimator of quantile $u_p$ and for the fixed value $k$, $r=n-k$ is a random variable. The asymptotic distribution of ordered impact index variables is defined:

 \begin{equation*}
\label{pronsysmpotsifi}
\mathbb{P}(\mathbf{X}_{(r)} \leq \tilde{u}_{pn} )=\mathbb{P}(\mathbf{X}_{(n-k)} \leq \tilde{u}_{pn} )=  \sum_{i=n-k}^n    \binom {n} {i} \mathcal{F}(\tilde{u}_{pn})^i (1-\mathcal{F}(\tilde{u}_{pn}))^{n-i}.
\end{equation*}

With using a result of [\cite{van2000asymptotic} chapter 21] we obtain Poisson approximation of the binomail distribution as follows: 
  \begin{equation}
\lim_{n \to \infty} \mathbb{P}(\mathbf{X}_{(n-k)} \leq \tilde{u}_{pn} )=  \sum_{j=o}^k  e^{-\lambda}\frac{\lambda^j}{j!},
\end{equation}

where  $$\lambda= \lim_{n \to \infty} n (1- \mathcal{F}(\tilde{u}_{pn}))$$.\\

The Poisson approximation can explain better since impact index variable is long tailed distributed and the variance is large.

\subsection{Classification Probability}

We can define a probability distribution over a set of classes of SIFIs and Non-SIFIs to obtain the class that a particular institution falls into. These probabilistic classifiers provide the information on the classification of institutions to predict the class of  each institution with a degree of certainty.\\

Let $(Y_{it})_{I \in \{1,\cdots,N \}, t \in {1,\cdots,T}}$ be ordered categorical variables expressed systemic rank of institutions $i$ at time $t$. The cumulative probability of these categorical variables with given set explanatory variables $x$ is expressed as:

\begin{eqnarray*}
\mathbb{P}(Y_{it} \leq j | x_{it}) = \mathcal{G} (\beta x_{it} ) , \quad j=1,\cdots ,J,
\end{eqnarray*}

where $J$ is the number of groups and function $\mathcal{G}$ is a cumulative distribution function. It is assumed the coefficient vector $\beta$ is not vary between the categories.\\

The explanatory variable $ x_{it}$ captures the information on features of institutions related size, complexity, and interconnectedness. In the section \eqref{SIFIsection}, we discussed briefly different factors to identify SIFIs.  All of the factors can be the significant effect on the probability of an institution becomes a SIFI. \\

The variable $\alpha_i$ contains time invariant information on individual characteristics of institutions. The probability of the institution $i$ at time point $t$ is assigned to systemic rink category $j$ is conditional on the explanatory variables $ x_{it}$ and the individual effect variable $\alpha_i$:

\begin{eqnarray*}
\label{probabiltylogit}
\mathbb{P}(Y_{it} =1 | x_{it} , \alpha_{i1} ) &=&  \mathcal{G} (\beta_1 x_{it} +\alpha_{i1} ) ,\\
\mathbb{P}(Y_{it} =  j | x_{it} , \alpha_{ij} ) &=& \mathcal{G} (\beta_j x_{it} +\alpha_{ij}  ) - \mathcal{G} (\beta_1 x_{it} +\alpha_{i1}  ), \quad  \quad j=2,\cdots ,J-1,\\
\mathbb{P}(Y_{it} =  J| x_{it} , \alpha_{iJ} ) &=&  \mathcal{G} (\beta_J x_{it} +\alpha_{iJ}  ).
\end{eqnarray*}

Note that the cumulative distribution function $\mathcal{G}$ can be expressed as a standard normal or logistic distribution function, and coefficient vector $\beta$ and $\alpha_i$ can be estimated by panel ordered probit or logit models.\\

\section{Example of Financial Network}
\label{SectExpFinNet}

In this section, we demonstrate how to identify the SIFIs by applying algorithm \ref{algo1} in a real interbank payment network. Since financial institutions' failures are atypical, and the corresponding data  are mostly confidential, we use a small sample dataset from \cite{Sorama2012} (FNA simulator platform: www.fna.fi).\\

The interbank payment system is the backbone for all financial transactions to facilitate the monetary  exchange between institutions. This system provides a set of procedures, laws, standards and requirements that can be unique for each financial system. An electronic payment system includes the flow of fund and the flow of information about counterparties and the value of transactions (for more information see: \cite{listfield1994modernizing}). \\

\begin{figure}[htbp]
\centering
\caption{\textbf{Example of Interbank Network of Payment Transaction}}
\includegraphics[scale=.45]{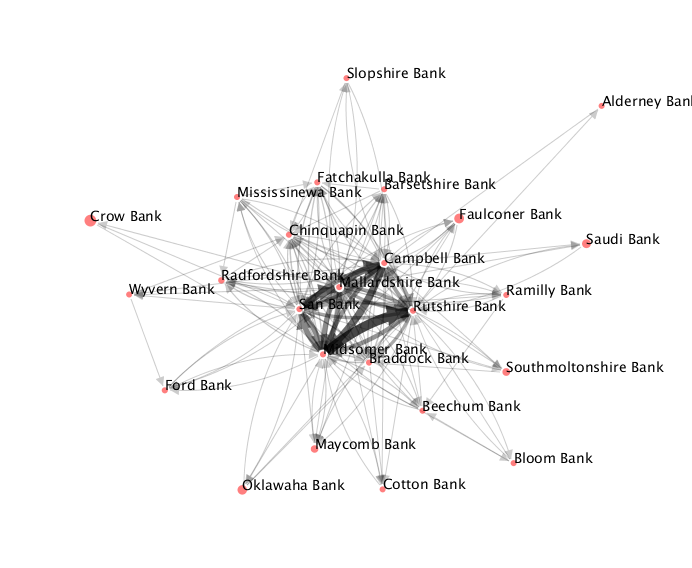} 
\captionsetup{width=.8\textwidth}
\captionsetup{font={small,it}}
\caption*{This figure shows an interbank payment network of payment transaction. The nodes are representatives of banks and links are representative of payment transactions. We created this plot with FNA simulator platform (www.fna.fi) }
\label{networkexample}
\end{figure}

We compare our result with SinkRank method introduced by \cite{Sorama2012} . Their method is based on  absorbing Markov chains with fixed  transition probability to model dynamics of liquidity cascade in a payment system. \cite{Sorama2012} used a payment system to identify banks which are the most affected by the failure. Table \ref{tabelresultExample} demonstrates results obtained from our algorithm (algorithm \ref{algo1}) and SinkRank method. Our algorithm uses slightly different assumptions, however, both method consistently can identify SIFIs in a network.\\

Results show that systemic risk index has a positive correlation between a number of in- and out-degrees of the network but not with banks size.  In addition, our algorithm can identify institutions which are largely connected with strong transmission channels to SIFIs.  In the past, the size of financial institution's balance sheet, like total capital or total financial transaction, has been used to rank the institutions in the system, however, the interconnectedness of the financial system can be a major reason for failing the whole of the system.\\

\begin{table}[htbp]
\caption{\textbf{Comparison of the Results of our Impact Index and SinkRank (\cite{Sorama2012})}}
\label{tabelresultExample}
\captionsetup{font={small,it}}
\centering
\scalebox{0.7}{
\begin{tabular}{clccccccccc}
\hline\hline
No&Name Bank&$Degree^+$&$Degree^-$&Impact Index&Inv-SinkRank&Total Received&Total Sent&Equity\\\hline
1&\textbf{San} &286&318&4.472092691&0.131266875&5.35E+09&6.13E+09&6.42E+09\\
2&\textbf{Mallardshire}&546&645&14.02179686&0.345049668&1.13E+10&1.37E+10&1.35E+10\\
3&Campbell&565&528&2.555336914&0.255536592&1.23E+10&1.11E+10&1.48E+10\\
4&\textbf{Midsomer}&413&331&4.327629834&0.154341717&8.41E+09&6.70E+09&1.01E+10\\
5&\textbf{Rutshire}&305&326&6.031221629&0.144775534&6.48E+09&7.11E+09&7.78E+09\\
6&\textbf{Braddock}&285&259&7.384959536&0.11383881&5.97E+09&5.27E+09&7.16E+09\\
7&Chinquapin&76&60&2.021948895&0.016499552&1.10E+09&8.24E+08&1.32E+09\\
8&Fatchakulla&25&21&0.896847037&0.004000575&2.59E+08&2.05E+08&3.11E+08\\
9&Radfordshire&20&28&0.275514539&0.006118699&2.30E+08&3.04E+08&2.76E+08\\
10&Barsetshire&9&13&0.141456433&0.002092251&7.31E+07&1.01E+08&8.77E+07\\
11&Beechum&10&8&0.128461665&0.001245738&7.41E+07&6.10E+07&8.89E+07\\
12&Ramilly&12&12&0.534263772&0.002795518&1.32E+08&1.37E+08&1.59E+08\\
13&Mississinewa&3&7&1.203555812&9.40E-04&2.33E+07&5.20E+07&2.80E+07\\
14&Faulconer&9&6&0.246758203&9.49E-04&6.45E+07&5.13E+07&7.74E+07\\
15&Wyvern&7&3&0.429629287&4.85E-04&7.78E+07&2.59E+07&9.33E+07\\
16&Southmoltonshire&11&6&0.157081723&9.17E-04&8.96E+07&4.16E+07&1.08E+08\\
17&Slopshire&1&1&0.209457126&5.78E-05&2743850.7&3611388.36&3.29E+06\\
18&Oklawaha&4&5&2.522957064&5.71E-04&2.03E+07&2.78E+07&2.43E+07\\
19&Cotton&3&6&0.829804475&9.38E-04&2.35E+07&4.85E+07&2.82E+07\\
20&Maycomb&2&1&0.036815624&8.76E-05&8173403.33&4175319.83&9.81E+06\\
21&Alderney&2&4&0.264143582&4.77E-04&1.06E+07&2.32E+07&1.27E+07\\
22&Ford&2&5&0.098661454&7.08E-04&1.37E+07&2.90E+07&1.65E+07\\
23&Crow&2&2&0.063812832&1.10E-04&8928062.78&6856694.69&1.07E+07\\
24&Saudi&1&1&0.178489716&6.04E-05&2876870.83&3505134.15&3.45E+06\\
25&Bloom&1&4&0.220222296&4.68E-04&6486514.2&2.23E+07&7.78E+06\\
\hline \hline
\end{tabular}%
}
\captionsetup{width=.8\textwidth}
\captionsetup{font={small,it}}
\caption*{ This table shows the information on total payment value received and sent, capital value and out-degree, in-degree  and ranking information of each bank. Results show that the distribution of Impact Index and SinkRank are relatively positive skew with 1.62 as the degree of asymmetry of mean and correlated with out-degrees and in-degrees. The banks, which are identified as systemic important banks, are highlighted. }
\end{table}%

\begin{figure}[htbp]
\caption{\textbf{Results of Applying Algorithm \ref{algo1} on Example Network Dataset}}
    \centering
    \begin{subfigure}[b]{0.45\textwidth}
                \includegraphics[width=\textwidth]{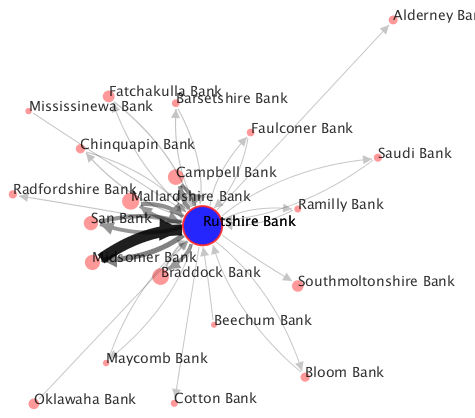}
        \label{fig:gull}
    \end{subfigure}%
    \begin{subfigure}[b]{0.5\textwidth} \includegraphics[width=\textwidth]{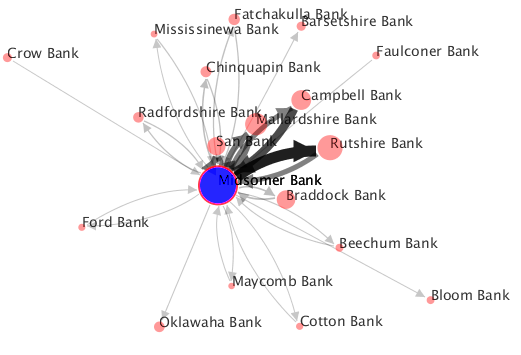}
        \label{fig:tiger}
    \end{subfigure}
    \captionsetup{font={small,it}}
\captionsetup{width=.8\textwidth}
      \caption*{This figures show results of applying algorithm \ref{algo1} on example network dataset indicate two systemically important banks; these banks are systemically important not just because of their balance sheet site but even they have interactions to most of nodes in the system. }\label{fig:resultexample}
 \end{figure}

\begin{figure}[htbp]
 \caption{\textbf{Results of Applying Algorithm \ref{algo1} and SinkRank Method (\cite{Sorama2012})on Example Network Dataset.}}
    \centering
    \begin{subfigure}[b]{0.5\textwidth}
                \includegraphics[width=\textwidth]{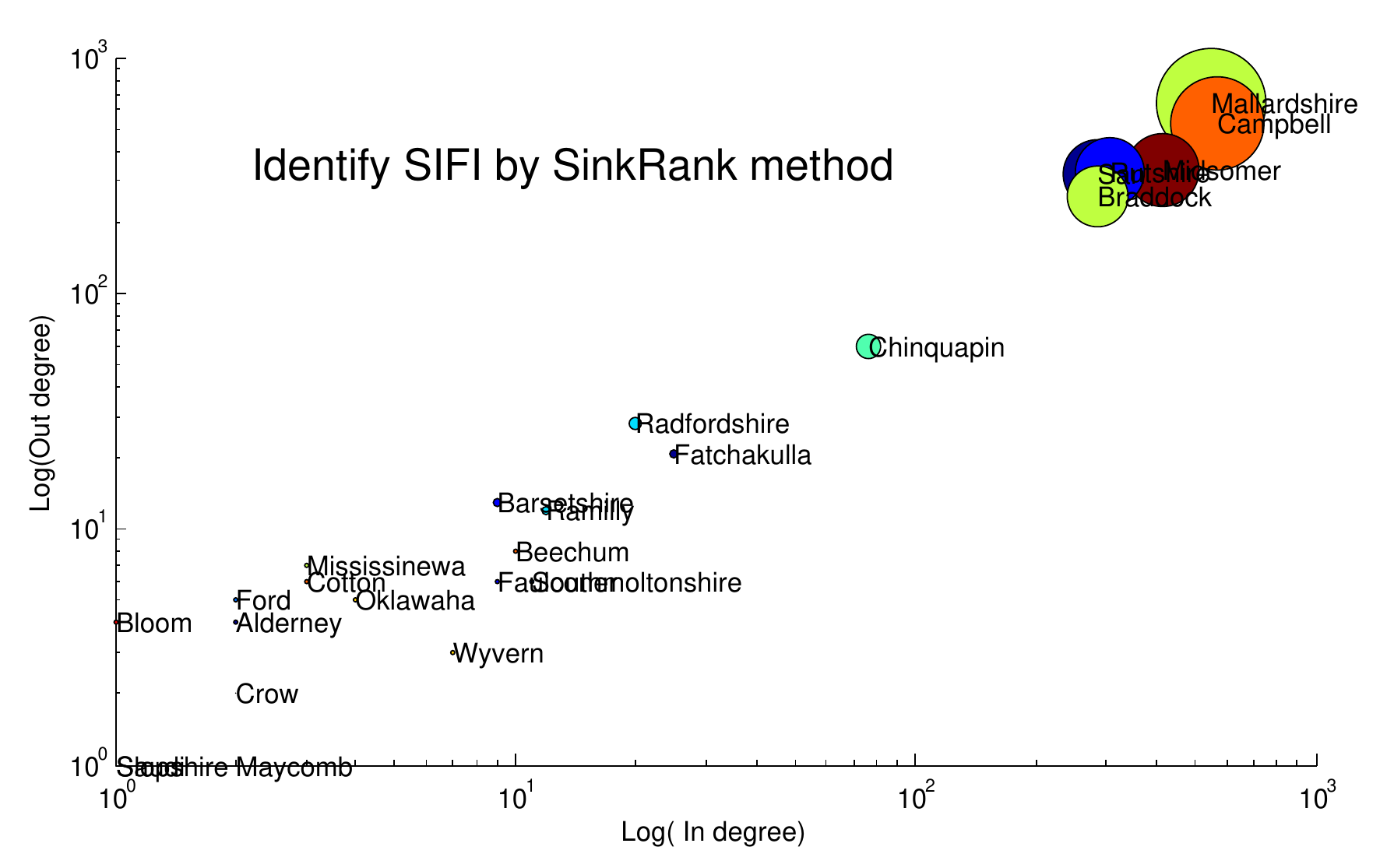}
                \caption{Identify SIFI by SinkRank method}
        \label{fig:gull}
    \end{subfigure}%
    \begin{subfigure}[b]{0.5\textwidth} \includegraphics[width=\textwidth]{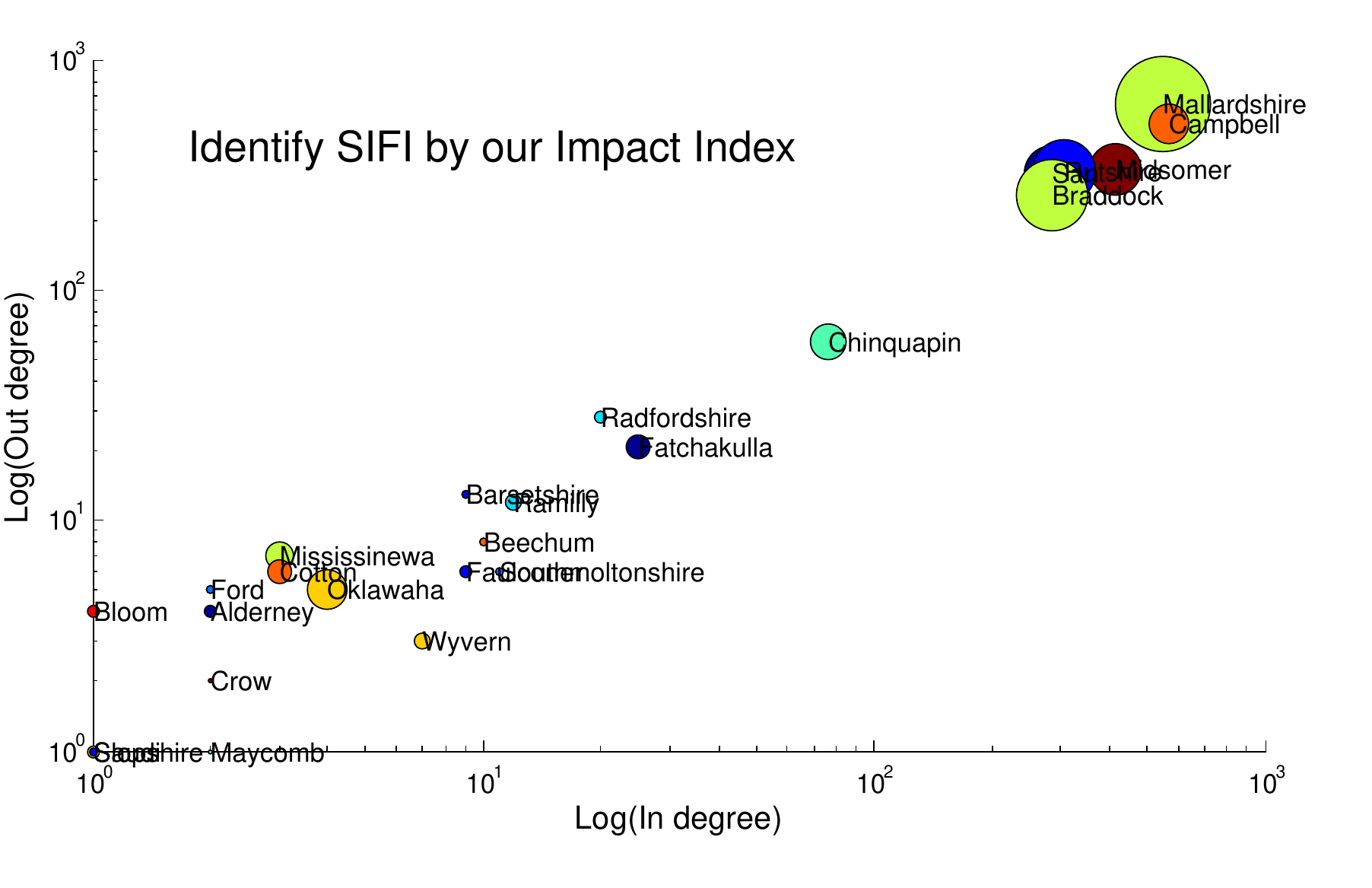}
        \label{fig:tiger}
    \caption{Identify SIFI by our Impact Index}
    \end{subfigure}
\captionsetup{width=.8\textwidth}
\captionsetup{font={small,it}}
      \caption*{This figures shows results of applying algorithm \ref{algo1} and SinkRank method on example network dataset. Axises are a total number of payments sent and received by banks and size of bubbles is represented of the magnitute of the systemic importance of banks calculated by both algorithms. }\label{fig:resultexample}
 \end{figure} 

%
%
%

\section{Numerical Experiments}
\label{SecnumbericExp}
In this section, we explain the numerical simulation of interbank network and compare our approximation solution with an optimum solution on simulated networks. Furthermore, we show the factors that lead to an insinuation become systemic important. These factors are the spotlight of regulations of SIFIs to reduce the probability of SIFIs' failures and enhance the financial stability. \\

As a summary: at time $t=0$, the network will be formed, and at time $t=1$ nodes will receive individual or aggregate deposit shocks. Then a default cascade will start and our algorithm try to identify SIFIs. The cascade continues until reach equilibrium and the algorithm might converge and find SIFIs.\\   

We create an interbank network as a direct random network. The structure of the desired network is a realization of the empirical studies which are explained in previous sections. We weight each link (financial relation) of the network with exposure magnitudes and each node (bank) with information on the balance sheet.  The total assets of a bank is the sum of its loans to other banks, as interbank assets, and the sum of its external assets. The liability of a bank includes the amount of loans taken from other banks, as interbank liabilities, and also its customer deposits. \cite{Gai2010} showed that the distributions of claims within the system are uniformly distributed as a result of intensity of banks to maximize their diversity lending strategy by distributing loans to all their debtors equally.\\ 

The simulation is performed in a way that ensures that the generated network represents a real-world banking system.  As mentioned earlier , empirical studies found that an interbank network has a fat-tailed degree. Exposures and in-degrees and out-degrees of the network follow Pareto distribution and power law distribution, respectably. This structure is stable across time. Because of the heterogeneous characteristics of the interbank liability (asset) structure, different interbanking systems are simulated. With given exposure, we estimate interbank assets and liabilities of banks. Equity is the proportion of one to two percent of risk-weighted assets. This proportion is randomly assigned to each bank. We balance each bank account with its deposit.
Figure \ref{figsimuDistri} presents distributions of simulated financial variables and  network characteristics. Financial variables capital, deposit and exposure have skewed (non-symmetric) distributions. Network properties size, out-degree, in-degree, betweenness, and closeness have non-symmetric distributions that are right skewed.

%
    \begin{figure}[htbp]
   \centering 
 \caption{\textbf{Histogram of Simulated Balance Sheet Data} \label{figsimuDistri}} 
  
\includegraphics[scale=.06]{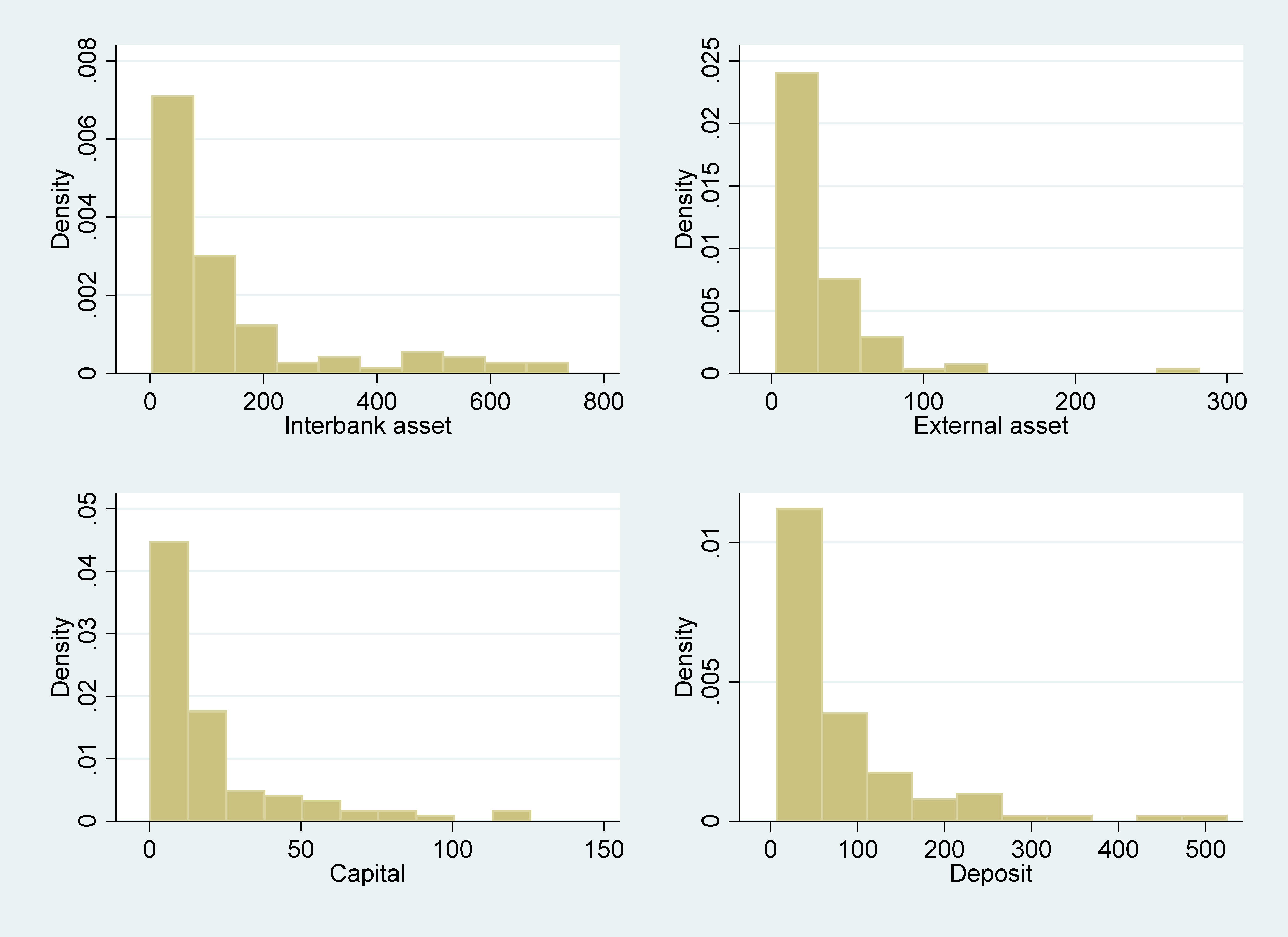} 
\centering
\captionsetup{font={small,it}}
\captionsetup{width=.6\textwidth}
\caption*{
 This figure plots the histogram simulated interbank assets, external assets, capital and deposit of institutions in the network. The variables are skewed distributed and they are correlated.  \label{FigsimulateConfiden}}
  \end{figure}%
  
    \begin{figure}[htbp]
    \centering 
    \caption{ \textbf{Histogram of Centrality Indices of Simulated Network}\label{FigsimulateConfiden}}
\includegraphics[scale=.06]{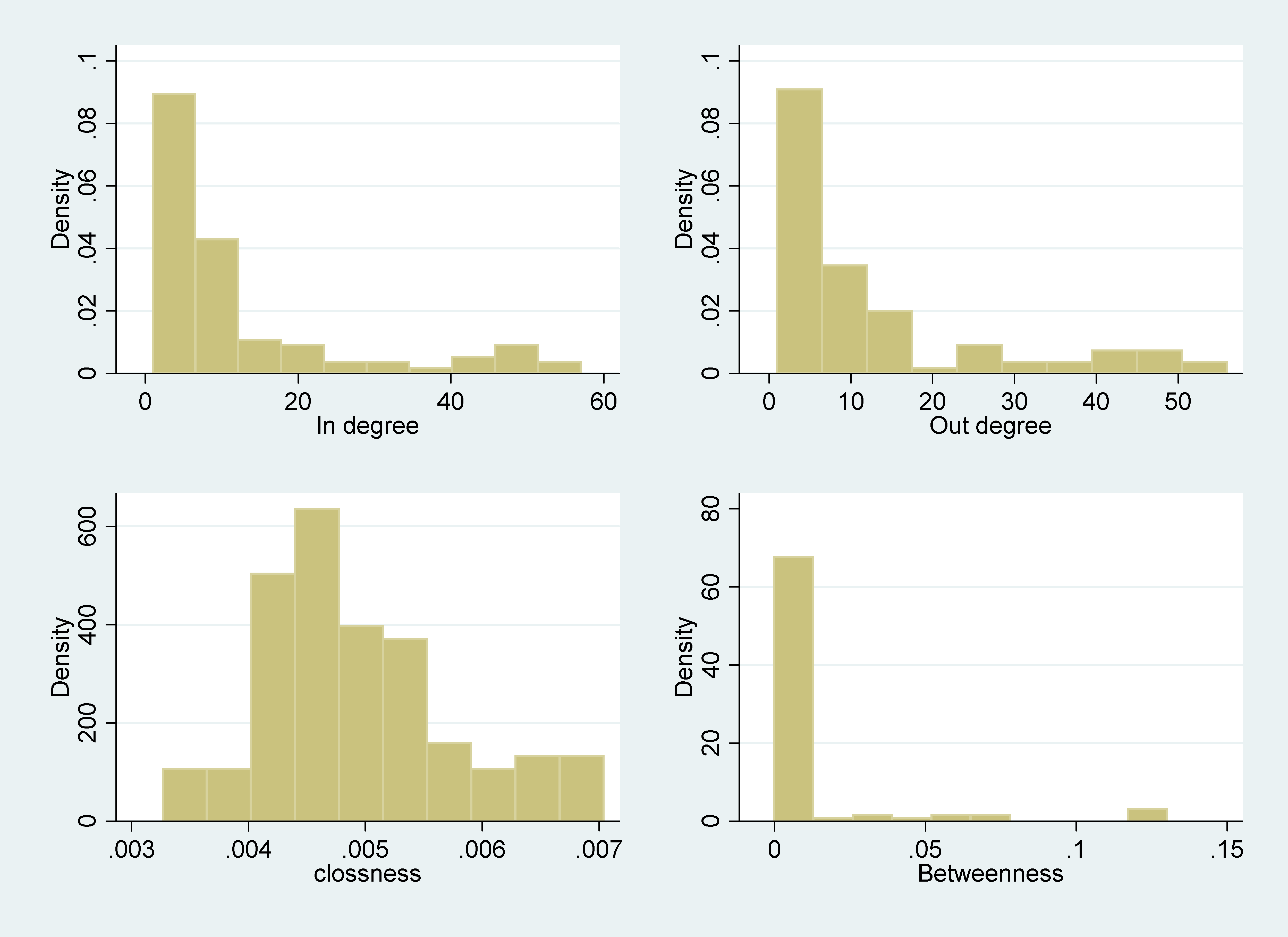} 
\centering
\captionsetup{font={small,it}}
 \captionsetup{width=.6\textwidth}
\caption*{  
The histograms of out-degree, in-degree, betweenness, and closeness are given. One can see that out-degree and in-degree are heavy-tailed distributed and they are positively correlated. Betweenness centrality and closeness centrality have low variation. }
    \end{figure}

From theoretical findings of \cite{allen2000financial}, if there exists a directed financial link from each bank to all others, banks spread their claims and diversify their contagion risk equally. We simulate a complete network to have an uniform perspective on a financial network and its asset-liability structures, (see Figure \ref{figcomplete}).\\

\begin{figure}[htbp]
\centering
\caption{ \textbf{Visualization of Complete Network}:\label{figcomplete} }
\includegraphics[scale=.35]{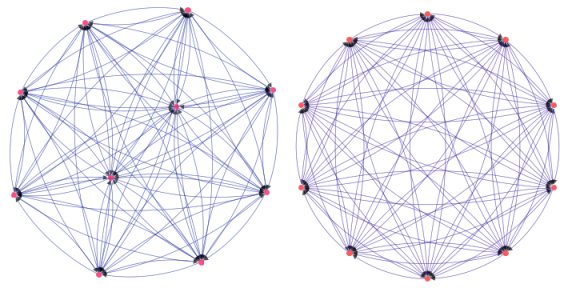} 
\center
\captionsetup{font={small,it}}
   \captionsetup{width=.6\textwidth}
\caption*{ Note: circle format (right side) and Cartesian format (left side). In a complete network there is a direct (indirect) connection between all pairs of nodes.}
\end{figure}

 \citep{allen2000financial} also indicated that less complete system increased the possibility of contagion risk due to highly interconnected links and lacks of potential for risk diversification. Past empirical studies showed that the full interconnected simulated system is an unrealistic structure, and we need to omit some directed edges within the simulated network.  Scale-free networks is one type of network with incomplete connectivity. \cite{Baraba2002} proposed the Barabasi-Albert (BA) algorithm for generating random scale-free networks. This algorithm starts with a set of small fully connected nodes and adds one node at a time with exactly k edges. Preferential attachment probability (proportion to the number of node's edges) is used to attach a new node to an existing node. This process continues until all nodes became connected. Under this procedure, it is possible for multiple edges to have a pair of nodes, and create loops in the network. In a typical financial system, institutions are serially connected in  a way that a preceding institution is connected to the foremost one. Figure \ref{figBAnetwork} shows a BA network in the circle and cartesian formats.\\

  \begin{figure}[htbp]
    \centering
 \caption{\textbf{Visualization of BA Network}} 
\includegraphics[scale=.35]{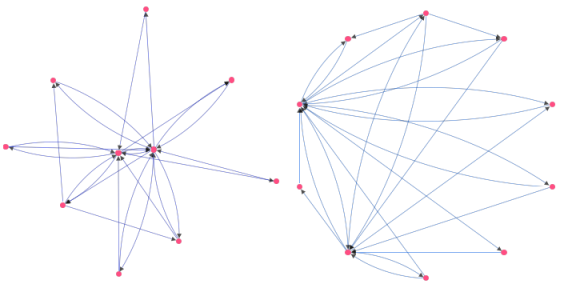} 
\center
\captionsetup{font={small,it}}
\captionsetup{width=.6\textwidth}
 \caption*{ Note : circle format (right side) and cartesian format (left side). This network is a scale-free network and the degree distribution of degrees follow a power law. \label{figBAnetwork}}
\end{figure}

\subsection{Numerical results}
\subsubsection{A comparison between approximation and optimal solutions}

In this section, we compare our approximation solution with an optimal solution of the problem \eqref{mainprob1}. The numerical results presented here were obtained by using Matlab's built on a PC running under processor (Intel Core 2.30 GHz). We firstly generate several random networks based on the \cite{Baraba2002} (BA) algorithm and also complete networks. Then, we simulate financial data for each bank with previously described distributions.\\

The optimal solution is obtained from built-in function (lsqlin) which solves constrained linear least square problem with interior point method. We create linear constraints of optimization problem according to financial relations in the network.\\ 
In order to compare the algorithm \ref{algo1} with the optimal algorithm (interior point method), the relative error is defined as follows:
\begin{equation}
Rel_{Err}= \frac{(Err_{Alg1}-Err_{interior})}{Err_{interior}},
\end{equation}
where $ Err_{Alg1}$ is the minimum objective value of algorithm \ref{algo1} and $ Err_{interior}$ is an objective value of an optimal solution.\\

Tables \ref{tableweek1} and \ref{tableweek2} present the numerical result obtained from the algorithm \ref{algo1} with  BA and complete network topologies while the number of nodes and edges are increased exponentially.\\

\begin{table}[htbp]
    \caption{\textbf{Comparison of Approximation and Optimal Methods}}
    \captionsetup{font=scriptsize}
    \begin{subtable}[v]{1\textwidth}
        \centering
        \scalebox{0.8}{
\begin{tabular}{c c c l c c l c c c l}
 \hline \hline
{Size} & &50&&&100&  && 200& \\ \hline
{Interco} && 2420&& &4821&  &&10475
\\
{Max in} &&129&&  &222& && 403 
\\
{Max out} &&135&&&  214 & & &394 
\\\hline
{Tol} &$10^{-3}$ &$10^{-4}$ &$10^{-5}$ & $10^{-3}$ & $10^{-4}$ & $10^{-5}$ & $10^{-3}$ & $10^{-4}$& $10^{-5}$
\\\hline
{Time} &$10^{-3}$ &$10^{-2}$ &0.015 & 0.0039 & 0.0075 & 0.026 & 0.019 & 0.055& 0.182
\\
{Ite} &18&58& 182 & 24 & 77 & 244 & 44 & 151 & 477
\\
{$Rel_{Err}$}& 0.034 & 0.043 & 0.050 &0.054 & 0.053 & 0.050 & 0.123 & 0.122 & 0.11
\\
\hline
\end{tabular}%
}
        \caption{Numerical results of BA network  \label{tableweek1}}
        
    \end{subtable}
    
    \hfill
     \begin{subtable}[v]{1\textwidth}
     
        \centering
        \scalebox{0.8}{
   \begin{tabular}{c c c l r c l r c c}
 \hline\hline
{Size} &&50&&&100& &&300 & \\ \hline
{Interco} &&9718&&&39474&&&359039& 
\\
{Max in} &&234&&& 443&&&1311& 
\\
{Max out} &&223&&& 453&&&1290&
\\\hline
{Tol} &$10^{-3}$&$10^{-4}$& $10^{-5}$&$10^{-3}$&$10^{-4}$& $10^{-5}$& $10^{-3}$& $10^{-4}$& $10^{-5}$
\\\hline
{Time} & 0.001&0.0013&0.0038& 0.002&0.005&0.018&0.028&0.095&0.304
\\
{Ite} &16&52&164&16&51&162&16&51&162
\\
{Rel-Err} &0.030 & 0.029 & 0.05 &0.050 & 0.049 & 0.045 & 0.062&0.060 &0.062
\\
\hline \hline
\end{tabular}%
}
        \caption{Numerical results of complete network  \label{tableweek2}}
    \end{subtable}
\captionsetup{width=.65\textwidth}
\captionsetup{font={small,it}}
    \caption*{Note: CPU time and relative error of algorithms \ref{algo1} toward the above-mentioned optimal solution are reported for both simulated networks of BA network and complete network. Size is number of nodes in the network, "Interco" is the total number of interconnections and "Max in", "Max out" are the maximum number of in-degree and out-degree,respectively. "Ite" is the number of iterations of the solution procedure to reach a tolerance value.  \label{tabcomparision}}
\end{table}

Computational results in Table \ref{tabcomparision} indicate that the algorithm \ref{algo1}  requires less CPU time for producing the much more accurate approximation solution to the optimal solution obtained by Matlab's built solver function. Moreover, CPU time does not increase exponentially when the size of the network grows. The results also show that the complexity of our algorithm does not depend on the number of interconnections, but on the number of nodes only. Thus, this algorithm can find an approximation solution to a large problem relatively close to the optimal solution with a very small error.\\

\subsection{Results of numerical experiments}

In this section, we present the results of simulations of the model and identify key players in the system and determine factors that cause the instability of the network. In the simulation, we generate a BA scale-free (directed) network with empirical properties of study Brazilian interbank networks by \cite{cont2012}. In each simulation, one to ten percent of networks' nodes received deposit shocks. These shocks could spread throughout the entire or a part of the system. These shocks also can be exogenous factors which have a strong influence on the entire system. Figure \ref{FigsimulateConfiden} shows the average of impact indexes of 100 simulations and its 95\% confidence interval. To have a robust estimation of impact index, we varied the percentages of nodes that received shocks in each simulation.\\

\begin{figure}[htbp]
\centering
\caption{ \textbf{Average of Impact Index and Corresponding 95\% Confidence Interval.} \label{FigsimulateConfiden}}
\includegraphics[scale=.65]{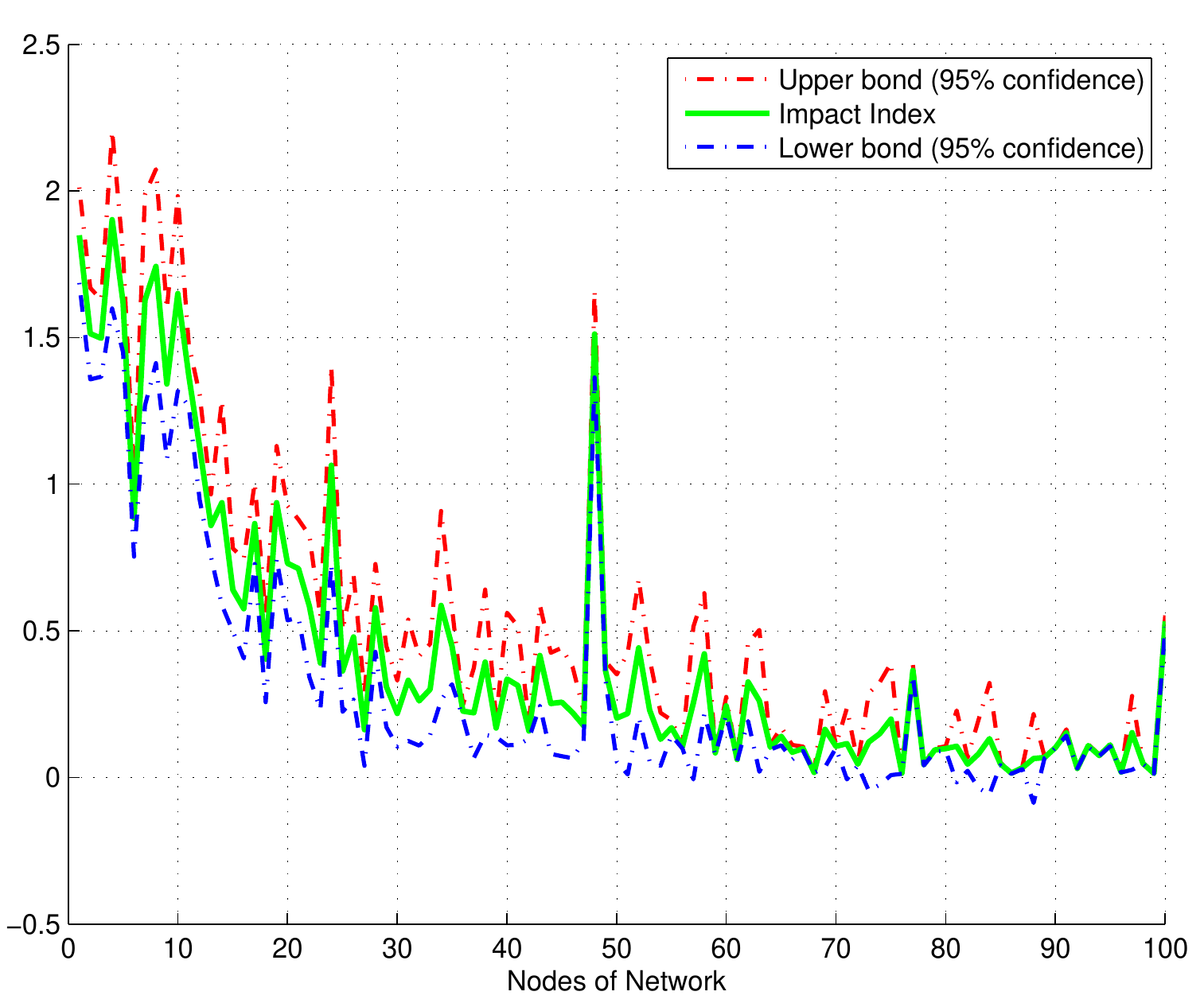} 
\centering
\captionsetup{font={small,it}}
   \captionsetup{width=.7\textwidth}
\caption*{The figure presents the average of estimation of impact index from 100 simulations (solid line) and 95\% confidence interval (dash line). In each simulation, a different number of nodes are hit by deposit shocks.}
\end{figure}

We hypothesize that impact indexes can be explained by means of the network which indicate the role and position of institutions in the financial system. To test this hypothesis, we regress the impact index on the balance sheet variables and centrality diameters of the network. Table \ref{tabtobitsim} presents the results of the impact of network properties (size, out-degree and in-degree) on measuring systemic risk. We use a Tobit model \footnote{We use Tobit model since the impact index is above zero and we have checked hypothesis of homoscedasticity and normality of residuals by  the Conditional Moment (CM) test (\cite{newey1985generalized}; \cite{tauchen1985diagnostic}).} to estimate the impact of capital, deposit, assets, and exposure on the impact index. Results from Tobit model and its marginal effects indicate an overall positive effects of changing capital and deposit on measuring systemic impact. The estimated coefficients capital, deposit and external assets are all positive and significant at 90\%, 95\% and 99\% levels, respectively.\\

In the next step, we add each network diameter to the base model (model 1). Based on Akaike Information Criterion (AIC) and Bayesian Information Criterion (BIC), we believe that network properties like degrees closeness and betweenness can explain the most variation of impact index and make other factors insignificant. From results pf the average marginal effect of models, we conclude that changing closeness has the highest average impact on impact index. Nevertheless, these network diameters can present the scope of influences of an institution on its counterparties, and at stress time, how much this institution can transmit the risk to the system. \\

\begin{table}[htbp]
\centering
\def\sym#1{\ifmmode^{#1}\else\(^{#1}\)\fi}
\caption{\textbf{Discriptive Statistics of Simulated Network} \label{DisStatiSimula} }

 \begin{subtable}[v]{1\textwidth}
        \centering
\scalebox{0.75}{
\begin{tabular} {@{} l l c c c c c @{}} \\ \hline
 & \textbf{Variables}  &\textbf{      Mean} & \textbf{ Std. Dev.} & \textbf{Min} & \textbf{Median} & \textbf{Max} \\
\hline
      Panel I: non Systemic  & internal asset&  100.1608 &    116.034 &   4.131724 &    64.2117 &   684.3691 \\
          &  external asset &  27.68229 &    26.9476 &   2.889826 &   17.94683 &   141.8025 \\
          &  capital  & 14.41384 &   17.29865 &   .3235944 &   8.162905 &   99.56005 \\
          & deposit  &  63.34887 &   56.16821 &   7.700143 &   47.20529 &   351.2374 \\
          &  in\_degree& 8.844444 &   9.459448 &          1 &        5.5 &         50 \\
          &  out\_degree&  8.766667 &   8.991944 &          1 &          5 &         42 \\
          & closeness &  .0047617 &   .0006452 &    .003268 &   .0046948 &   .0064103 \\
          &   betweenness & .0026189 &   .0102355 &          0 &          0 &   .0664723 \\
                    &    &  &    &          &           &    \\
            Panel II:   systemic   &  internal asset& 496.209 &   181.0448 &   74.98804 &   532.3312 &   737.1625 \\
          & external asset &  60.44734 &   80.09678 &   14.64383 &     30.266 &   281.8516 \\
          &capital  &   63.03347 &   38.32387 &   13.44478 &   65.30916 &   125.8125 \\
          & deposit &  245.5183 &   142.2358 &   38.69307 &   221.2899 &   524.7476 \\
          &   out\_degree&     42.4 &   13.89804 &          8 &         46 &         57 \\
          &   in\_degree&    43.1 &   13.11869 &          9 &         47 &         56 \\
          &  closeness &  .006456 &   .0005866 &   .0049261 &   .0066228 &   .0070423 \\
          & betweenness &  .0719681 &   .0506624 &    .000404 &   .0666043 &   .1301973 \\
             &   &  &   &    &    &    \\
  Panel III: total  &   internal asset& 139.7656 &   171.2983 &   4.131724 &   72.89974 &   737.1625 \\
          &  external asset &   30.9588 &    36.5191 &   2.889826 &   21.32403 &   281.8516 \\
          &  capital  &  19.2758 &   24.84822 &   .3235944 &   11.28302 &   125.8125 \\
          &  deposit& 81.56582 &   87.70555 &   7.700143 &   51.16003 &   524.7476 \\
          &    in\_degree&    12.2 &   14.15499 &          1 &        6.5 &         57 \\
          &   out\_degree&     12.2 &   13.98195 &          1 &        6.5 &         56 \\
          & closeness &   .0049312 &   .0008164 &    .003268 &   .0047619 &   .0070423 \\
          & betweenness &   .0095538 &   .0276537 &          0 &          0 &   .1301973 \\
\hline

\end{tabular}
}
\end{subtable}
\begin{subtable}[v]{1\textwidth}
\captionsetup{justification=raggedleft,singlelinecheck=false}
        \centering
\caption*{Tests of differences of mean of systemic and non systemic (simulated network)}
\centering
 \scalebox{0.8}{
\begin{tabular}{lrr}
\hline 
Variable & t-test & Wilcoxon-Mann-Whitney test \\ 
\hline \hline
internal asset& -6.7650&-4.550 \\
external asset &  -1.2855&-2.149 \\
capital& -3.9672& -4.205\\
deposit& -4.0155&-4.240 \\
in\_degree&  -7.4458 &-4.638 \\
out\_degree& -8.0682&  -4.814\\
closeness& -8.5746&-4.757 \\
betweenness& -4.3189 &-5.647 \\
\hline

\end{tabular}
}
\end{subtable}
\captionsetup{font={small,it}}
   \captionsetup{width=.7\textwidth}  
    \caption*{The above table presents the discriptive statistics of simulated network for both systemic and non-systemic members. One can see the differences of network indicators of systemic and non-systemic nodes. The below table confirms the significant differences of means of network properties which are tested by t-test and Wilcoxon-Mann-Whitney test.}
\end{table}

\begin{table}[htbp]
\caption{\textbf{Impact Index Regression on the Information of Balance Sheet and Node Centrality in Networks} \label{tabtobitsim}}
\centering
\def\sym#1{\ifmmode^{#1}\else\(^{#1}\)\fi}
\begin{subtable}[v]{1\textwidth}
\centering
\caption*{\textbf{Coefficients of Tobit Model of Impact Index}}
\scalebox{0.8}{
\begin{tabular}{l*{5}{c}}
\hline\hline
       
            &\multicolumn{1}{c}{Model 1}&\multicolumn{1}{c}{Model 2}&\multicolumn{1}{c}{Model 3}&\multicolumn{1}{c}{Model 4}&\multicolumn{1}{c}{Model 5}\\
\hline
model       &                     &                     &                     &                     &                     \\
Ln(capital)   &      0.0818\sym{**} &       1.055         &      0.0367         &      0.0469         &     -0.0124         \\
            &    (0.0280)         &    (0.0305)         &    (0.0267)         &    (0.0316)         &    (0.0485)         \\
Ln(deposit)   &       0.363\sym{***}&       1.188         &       0.166         &       0.195\sym{*}  &      0.0957         \\
            &    (0.0885)         &     (0.108)         &    (0.0840)         &    (0.0882)         &     (0.170)         \\
Ln(external asset)&      -0.127\sym{*}  &       0.962         &     -0.0384         &     -0.0548         &     -0.0215         \\
            &    (0.0554)         &    (0.0514)         &    (0.0502)         &    (0.0468)         &    (0.0720)         \\
Ln(exposure)  &      0.0452         &       0.955         &     -0.0153         &     -0.0414         &       0.253\sym{*}  \\
            &    (0.0475)         &    (0.0439)         &    (0.0403)         &    (0.0402)         &    (0.0937)         \\
Ln(in-degree)     &                     &       1.302\sym{***}&                     &                     &                     \\
            &                     &    (0.0740)         &                     &                     &                     \\
Ln(out-degree)    &                     &                     &       0.251\sym{***}&                     &                     \\
            &                     &                     &    (0.0298)         &                     &                     \\
Ln(clossness) &                     &                     &                     &       1.773\sym{***}&                     \\
            &                     &                     &                     &     (0.370)         &                     \\
Ln(betweenness)&                     &                     &                     &                     &      0.0932\sym{**} \\
            &                     &                     &                     &                     &    (0.0280)         \\
\_cons      &      -0.970\sym{***}&       0.537\sym{***}&      -0.627\sym{***}&       9.291\sym{***}&     -0.0469         \\
            &     (0.175)         &    (0.0894)         &     (0.164)         &     (2.118)         &     (0.632)         \\
\hline
sigma       &                     &                     &                     &                     &                     \\
\_cons      &       0.298\sym{***}&       1.321\sym{***}&       0.264\sym{***}&       0.262\sym{***}&       0.304\sym{***}\\
            &    (0.0257)         &    (0.0340)         &    (0.0251)         &    (0.0287)         &    (0.0544)         \\
\hline
\(N\)       &         100         &         100         &         100         &         100         &          36         \\
AIC       &          53.53    &       42.18   &        31.49          & 29.91   &        30.48  \\ 
BIC        &         69.16   &        60.42   &        49.73          & 48.14    &       41.57\\

\hline\hline
\multicolumn{6}{l}{\footnotesize Standard errors in parentheses}\\
\multicolumn{6}{l}{\footnotesize \sym{*} \(p<0.05\), \sym{**} \(p<0.01\), \sym{***} \(p<0.001\)}\\
\end{tabular}
}
  \end{subtable}
    
    \hfill
     \begin{subtable}[v]{1\textwidth}
\centering
\def\sym#1{\ifmmode^{#1}\else\(^{#1}\)\fi}
\caption*{\textbf{ Average Marginal Effects of Tobit Model of Impact Index}}
        \scalebox{0.8}{
\begin{tabular}{l*{6}{c}}
\hline\hline
          &\multicolumn{1}{c}{Model 1}&\multicolumn{1}{c}{Model 2}&\multicolumn{1}{c}{Model 3}&\multicolumn{1}{c}{Model 4}&\multicolumn{1}{c}{Model 5}\\
\hline
Ln(capital)    &       1.059\sym{**} &      0.0377         &      0.0261         &      0.0332         &     -0.0107         \\
                &    (0.0208)         &    (0.0205)         &    (0.0190)         &    (0.0225)         &    (0.0421)         \\
Ln(deposit) &       1.288\sym{***}&       0.121         &       0.118\sym{*}  &       0.138\sym{*}  &      0.0829         \\
                 &    (0.0802)         &    (0.0640)         &    (0.0601)         &    (0.0633)         &     (0.147)         \\
Ln(external asset)&     0.915\sym{*}  &     -0.0274         &     -0.0273         &     -0.0389         &     -0.0187         \\
                     &    (0.0356)         &    (0.0377)         &    (0.0358)         &    (0.0335)         &    (0.0623)         \\
Ln(exposure)         &       1.032         &     -0.0320         &     -0.0108         &     -0.0293         &       0.220\sym{**} \\
                   &    (0.0339)         &    (0.0323)         &    (0.0287)         &    (0.0287)         &    (0.0813)         \\
Ln(in-degree)           &                     &       0.185\sym{***}&                     &                     &                     \\
                    &                     &    (0.0386)         &                     &                     &                     \\
Ln(out-degree)                &                     &                     &       0.178\sym{***}&                     &                     \\
                            &                     &                     &    (0.0195)         &                     &                     \\
Ln(clossness)             &                     &                     &                     &       1.257\sym{***}&                     \\
                              &                     &                     &                     &     (0.256)         &                     \\
Ln(betweenness)                 &                     &                     &              &       &      0.0807\sym{**} \\
                                 &                     &                     &                     &                     &    (0.0250)         \\
\hline
\(N\)                   &         100         &         100         &         100         &         100         &          36         \\
\hline\hline
\multicolumn{6}{l}{\footnotesize Standard errors in parentheses}\\
\multicolumn{6}{l}{\footnotesize \sym{*} \(p<0.05\), \sym{**} \(p<0.01\), \sym{***} \(p<0.001\)}\\
\end{tabular}
}
    \end{subtable}
\captionsetup{font={small,it}}
   \captionsetup{width=.7\textwidth}  
    \caption*{{\small Note: The dependent variable is impact index with above zero value. Thus, we need to use a censored regression model like Tobit model. To satisfy the primary condition of Tobit model, we use natural logarithm of independent variables. In Tobit model, since the latent variable can be extended below zero, the  interpreting the effect of independent variables (coefficients) is not valid instead we need to estimate marginal effects which can explain the effects while the latent variable being at/or above zero.}}
\end{table}

Generally, these indicators of centrality are correlated and contain common information, and low correlation between them indicates valuable information about network formation. Generally, SIFIs should have high value in terms of network centralities.  Figures \ref{clossBetween}, \ref{degreclossness} and \ref{degreBetween} illustrate the relation between network centralities.  An institution with the high-value degree and low-value betweenness is considered as an institution with a large number of counterparties and the ability to spread the stress with high speed.
 Institutions with the high-value degree and low-value closeness might root in a group which is far away from the rest of the network. High-value closeness and low-value betweenness indicates institutions with having multi-layers of counterparties and multiple connections in the network. These institutions have a high impact on their neighbors. Finally, the institutions that have high-value betweenness and low-value closeness try to monopolize the relations from a small number of counterparties to a large number.\\

\begin{figure}[htbp]
    \centering
\caption{\textbf{Relationship Between Centrality Degrees
 \label{centralityplot}}}
    \begin{subfigure}[htbp]{.5\columnwidth}
\includegraphics[scale=.35]{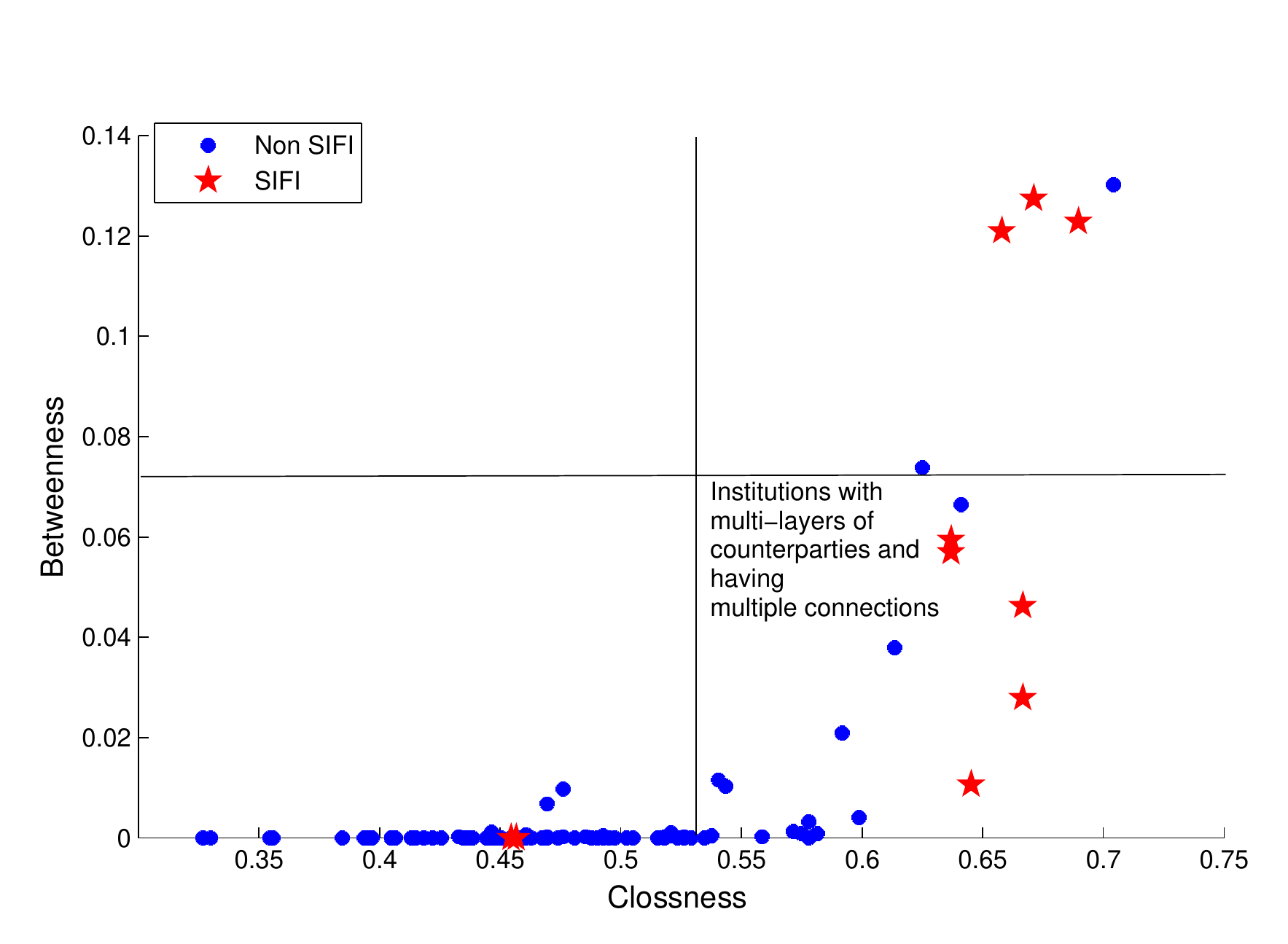} 
\centering
\caption{\textbf{Clossness versus Betweenness} \label{clossBetween}}
  \end{subfigure}%
    \begin{subfigure}[htbp]{.5\columnwidth}
\includegraphics[scale=.35]{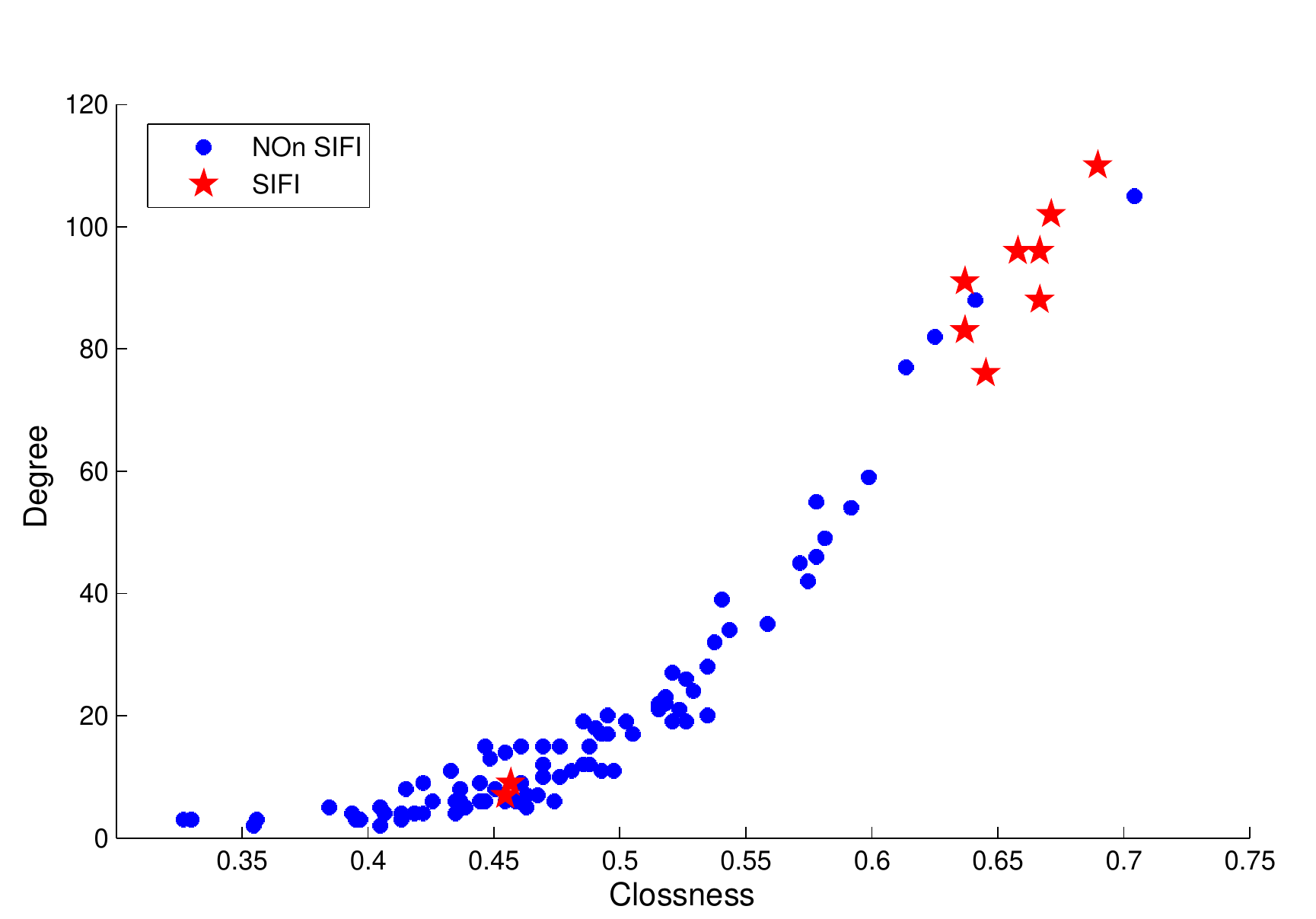} 
\centering
\caption{\textbf{Degree versus Clossness} \label{degreclossness}}
    \end{subfigure}
        \begin{subfigure}[htbp]{.5\columnwidth}
\includegraphics[scale=.35]{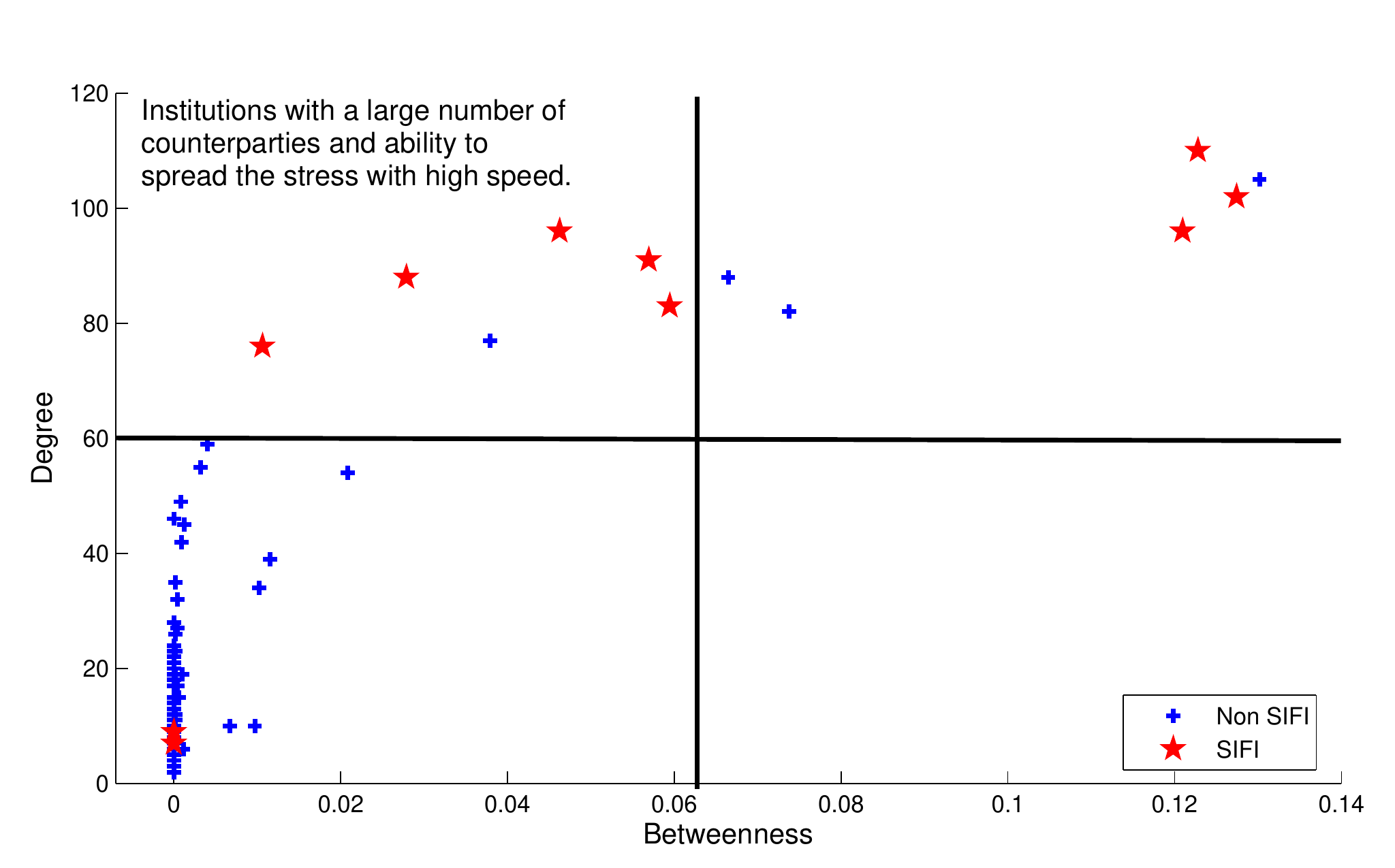} 
\centering
\caption{\textbf{Degree versus Betweenness  } \label{degreBetween}}
    \end{subfigure}
    \captionsetup{font={small,it}}
    \captionsetup{width=.7\textwidth}
    \caption*{The figure plots the scatter plots of centrality degrees of network nodes. The SIFIs are shown by the star symbols. Generally, these centrality degrees are correlated and contain common information about connections of nodes.}
\end{figure}

After we have determined the factors that lead to an institution became an SIFI, we estimate the sensitivity of these factors on the change of the systemic risk level of institutions which are considered as SIFI.  Another form of presentation of this risk is odds-ratio which reflects the likelihood that an institution is classified as an SIFI. Here, we use log-odds ratio, \cite{greene2003econometric} showed that a change in the log-odds ratio of an event can be interpreted as a change in the probability of an occurrence of that event.\\

We estimate the probability of an institution to be ranked as an SIFI from the equation:  \eqref{probsifi}.  Figure \ref{figprobSIFI} shows that probability density function that institutions are classified as SIFIs. This figure also shows that changing this probability with changing the value of central degree indicators of SIFIs. With defining quantiles equation \eqref{quantileEq} , we classify the institutions in three groups and construct 95\% confidence for each classification border. Figure \ref{figoddsSIFI} presents the corresponding the log odds ratios. We test increasing or decreasing the log-odds ratio of risk of an institution to be ranked as an SIFI by changing its value of central degree indicators. Results of tests are given in Table \ref{tabchangodds}. Overall results confirm that the significant decrease (-1.685401) of the average of log odds by decreasing in-degree. We also observe an significant increase log-odds ratio (1.52049 and 0.9098) in average by decreasing out-degree and closeness. The 99\% confidence intervals are given for difference average of odds ratios. These evidences confirms our hypothesis on the role of interconnection in stress propagation and amplification.

\begin{table}[htbp]\centering
\def\sym#1{\ifmmode^{#1}\else\(^{#1}\)\fi}
\caption{ \textbf{Statistical Test of Changing Odds Ratios} \label{tabchangodds}}
        \scalebox{0.65}{
\begin{tabular}{ccccccccc}
\hline 
\hline 
Variable &	 Obs     &   Mean (diff) &  Std. Err.   &Std. Dev. &t-statistic & \multirow{2}{20mm}{Wilcoxon z-statistic}   & \multicolumn{2}{c}{[99 \% Conf. Interval]}\\
 &   &  &  & &&   & &\\
\hline 
 $\Delta$log(odds\_changeindegre) &	 48  & -1.685401\sym{***} &   0.1533723    &1.062594 & -10.9890& -6.031& -1.993946  &   -1.376856\\
$\Delta$log(odds\_changeoutdegre)	& 62    &1.520497\sym{***} &    0.1223435&    0.9633334  & 12.4281&  6.846& 1.275856  &  1.765137\\
 $\Delta$log(odds\_changeclossness)	& 62 &   0.9098212\sym{***}    & .07439545  &  0.5857904  & 12.2295& 6.846 & .7610583 &   1.058584\\
$\Delta$log(odds\_changebetween) &	62 &  -.00218055\sym{***}   &  .0001791  &  .00141021& -12.1753 & -6.846&  -.0025387&   -.0018224\\
\hline 
\multicolumn{6}{l}{\footnotesize \sym{*} \(p<0.05\), \sym{**} \(p<0.01\), \sym{***} \(p<0.001\)}\\
\end{tabular}
}
\captionsetup{font={small,it}}
 \captionsetup{width=.78\textwidth}
\caption*{Table reports results of sensitivity the logarithm of odds ratios by increasing centrality indicators. We use t-test and the Wilcoxon signed-rank test as a non-parametric statistical hypothesis test. We compare two paired samples  to assess their population average. Mean(diff) is an average of difference of two samples. The corresponding 99\% confidence intervals are given.}
\end{table}

\begin{figure}[htbp]
    \centering
\caption{\textbf{Relationship Between Centrality Degrees on Nodes in the Network
 \label{centralityplot}}}
    \begin{subfigure}[htbp]{.5\columnwidth}
\includegraphics[scale=.35]{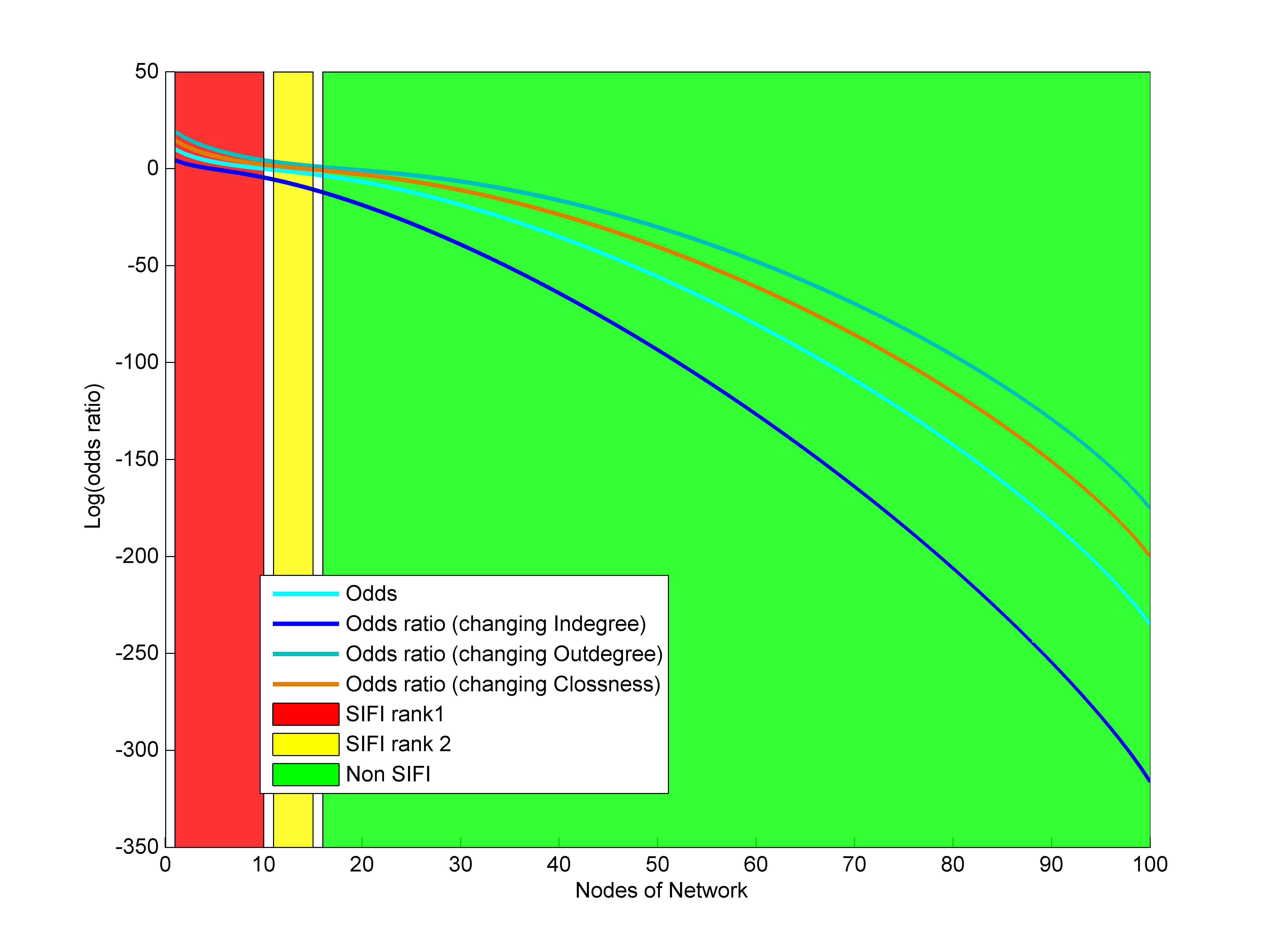} 
\centering
    \captionsetup{font={small,it}}
     \captionsetup{width=.75\textwidth}
\caption{Changing logarithm of Odds ratios of ranking institutions as SIFIs \label{figoddsSIFI}} 
  \end{subfigure}%
    \begin{subfigure}[htbp]{.5\columnwidth}
\includegraphics[scale=.35]{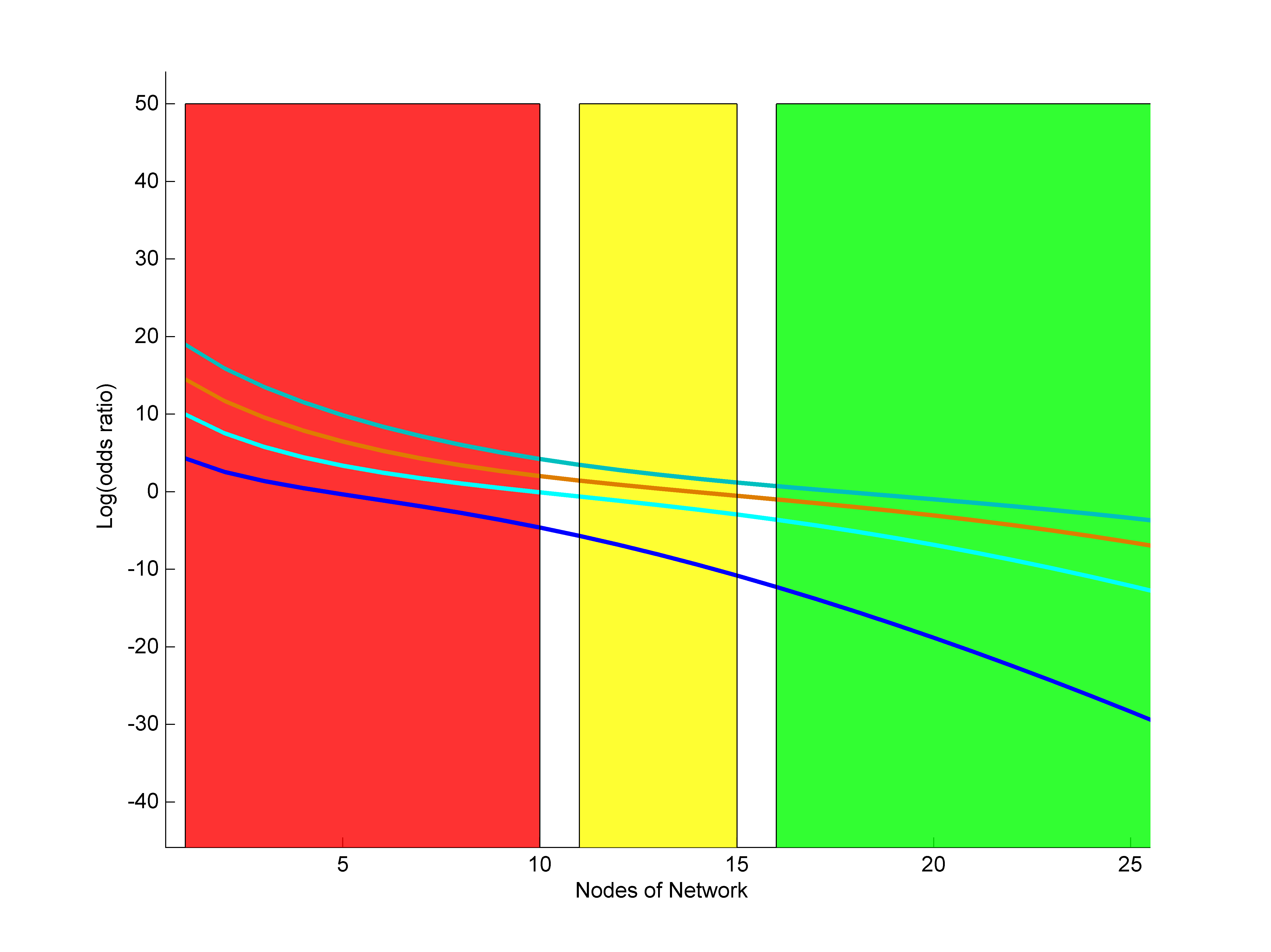} 
\centering
\caption{Changing logarithm of odds ratios of ranking the first 25 institutions as SIFIs \label{figoddsSIFI}}
    \end{subfigure}
        \begin{subfigure}[htbp]{.9\columnwidth}
\includegraphics[scale=.65]{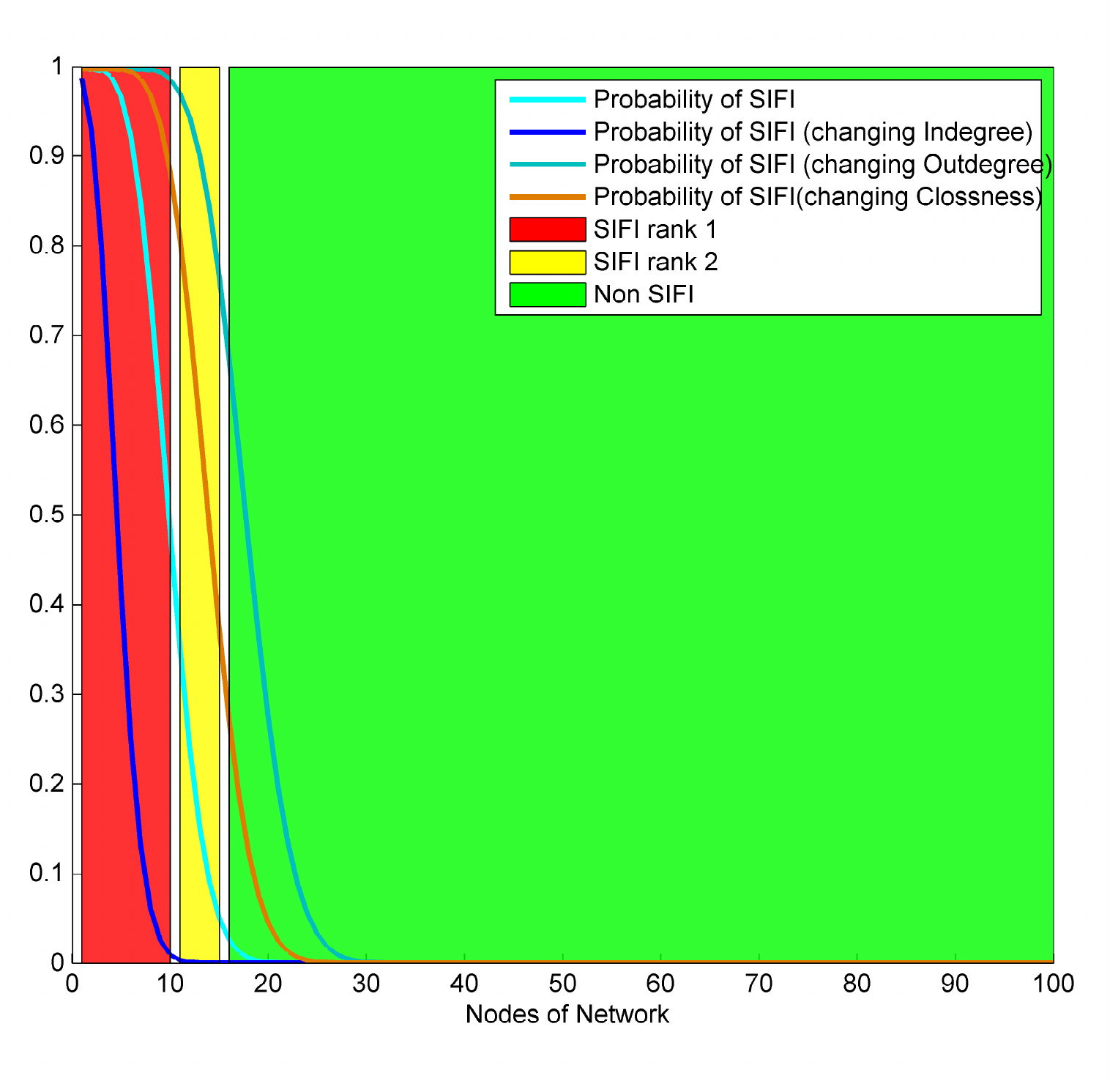} 
\centering
\caption{ Probability of ranking institutions as SIFIs\label{figprobSIFI}}
    \end{subfigure}
    \captionsetup{font={small,it}}
     \captionsetup{width=.78\textwidth}
    \caption*{The figure plot the probability and natural logarithm of the odds ratio of institutions might be considered as SIFIs. This probability would be changed if the centrality  degrees of nodes are changed. The red area represents the nodes with the highest probability ranked as SIFI and classified as the first SIFI bucket. The yellow area represents the second class of SIFIs. The none systemic institutions are located in the green area. The upper-right figure presets the natural logarithm of odds ratios in a bigger size.}
\end{figure}


\section{Empirical Study}
\label{SecEmprical}
In this section, on the empirical side, we investigate factors which lead to identifying SIFIs in a financial system. As stated earlier, the BCBS uses an indicator-based method and comprises five groups of features and characteristics of banks in determining the G-SIB. Banks should publicly disclose these metrics on the annual basis. These indicators are used to rank banks and calculate capital adequacy ratio, minimum requirements, for each bank. The goal is then to classify banks in terms of influence on stability of the system. Within these classes, the same level of regulation, regardless of the internal ranking, is assigned.\\

To asses the theoretical model, an empirical model is conducted. We use a sample of 221 the largest banks in each country. This global sample contains the information on financial aspects of banks which are chosen by BCBS to identify G-SIB candidates. Our sample includes banks that are currently considered systemically important, and banks that are or have been considered systemically important in the past. The sample of banks are divided into five categories: a group of banks eventually were not considered to be G-SIBs, and four groups of systemically important banks classified based different loss  absorbency obligation according to the Basel III. Table \ref{tebfrqSIFI} summaries the information on the minimum capital requirement of these categories: the magnitude of the capital requirement for the lowest bucket is 1.0\% of risk-weighted assets and the highest bucket is 3.5\%.\\

\begin{table}[htbp]
 \centering
 \def\sym#1{\ifmmode^{#1}\else\(^{#1}\)\fi}
\caption{\textbf{ Classification on Global-SIFI}\label{tebfrqSIFI}}
        \scalebox{0.65}{
\begin{tabular}{c|cccccc|c}
\hline\hline
        & Non systemic & Bucket 1 &Bucket 2  &        Bucket 3       &   Bucket 4    &  Bucket 5 &     Total\\
\makecell{Common Equity Tier 1 capital\\
 (\% of risk-weighted assets)}& &1 \%&1.5 \%&2 \%&2.5 \%&3.5 \%&\\\hline

         Year 2012 &     193     &    15      &    7         & 3   &       3& 0&221\\
        Year 2013 &    193    &     15       &   7    &      4   &       2& 0& 221\\ 
        Year 2014&   193  &       16 &         6  &        4         & 2& 0& 221\\
       Year  2015&  193       &  17  &        5  &        4         & 2 & 0& 221\\
\hline
     Total&       772  &       63  &       25    &     15 &    9&  0&    884 \\\hline \hline
\end{tabular}
}
\captionsetup{font={small,it}}
 \captionsetup{width=.8\textwidth}
\caption*{{\small The number of banks in the five categories of systemic important banks of 221 banks in four years. The corresponding capital requirement as the proportion of risk-weighted assets is given. The majority of banks are classified in bucket one with 1\% capital requirement. Changing the size of each bucket through years is due to updating the annual list SIFI by BCBS.}}
\end{table}%

These Five indicators are applied in determining the G-SIB by BCBS:\\

\begin{enumerate}[I]
\item
\textbf{Size:} it is the main measurement of systemic importance and can be quantified based on the overall exposure of an institution. It also indicates the magnitude of potential distress or failure of an institution. Measuring the size is dependent on assumptions of the methodology used and management's point of view. Basel III regulation \cite{basel2010basel} specified the measuring of the size, which is used for the leverage ratio, as the total amount of exposures of a bank scaled by the summation of the total amount of exposures of all banks in the sample. In our sample the total amount of exposures was not available for all banks, instead, we use the total amount of assets. This variable is important to classify a bank as a systemic important and we expect this variable to have a positive influence on increasing probability of a bank becoming a systemically important bank.\\

\item
\textbf{Interconnectedness:} it measures the interrelation of banks in a financial system and how the resources are distributed among players in the system. The recent financial crisis indicates interconnectedness plays a key role in the transmission of risk and can cause a cascading failure of counterparties. We can measure this indicator based on shares of bank asset and liabilities in the financial system. In our sample, we have information on interbank ratios and the amount of deposits from other banks, and we believe it has a positive effect on the probability of bank to be ranked as a SIB.\\

\item
\textbf{Substitutability/financial institution infrastructure:} this indicator contains information on the importance role of a financial institution to provide fundamental services in a particular business line. It also can show unknown substitution costs of a financial institution. One can measure the financial institution infrastructure of a bank with a number of payment activities and the amount of  financed transactions in liability and equity markets. We measure this indicator with a number of loans that a bank can provide for its customers. In the case of failure of the bank, customers need to look for the similar service from an alternative institution.\\

\item
\textbf{Complexity:} this indicator measures the functional complexity of a financial institution. The higher level of complexity requires more time and cost to resolve a financial institution. One can measure the level of complexity with an amount of OTC derivatives, available-for-sale securities, and level 3 assets. We measure this indicator with the amount available for sale securities, liquid assets funding, and trading assets since these factors can open channel for a bank for fire sale which is one of a significant dimension of  systemic risk. We can have a hypothesis of a positive impact of this indicator on the desired probability.\\

\item
\textbf{Cross-jurisdictional activity:} this indicator explains relatively the global financial activities of a bank and measures the international consequences of a bank's distress or failure. We calculate this indicator as a propotion of the total asset of a bank in the foreign market of the country where the bank is located. The foreign market index is measured by the total value of exports of goods and services of a country as a percentage of its GDP on 1-7 scale (see: \cite{schwab2015world}). 

\end{enumerate}

Table \ref{statdis} summarizes the variables which are used in our empirical study and Figure \ref{figboxplots} illustrates the boxplots of main variables. We obtain bank-level data from Bankscope database. Panels I and II summarize the information of non-systemic (bucket 0) and systemic (buckets 1-5) banks , respectively. The total asset was an indicator of the size of a bank and it was used to rank banks. From the result of the statistical test and an illustration of boxplot, we can observe the total assets of systemic banks is decreasing over years.  Interbank ratio is calculated as the proportion of due from counterparties to due to counterparties and indicates the interbank liquidity. We can observe that the systemic banks are funds providers rather than funds borrowers. Meanwhile, these banks try to level this ratio over the time. The liquid asset is calculated by a total of cash, trading securities and outstanding from other banks and central banks. Results show that systemic banks are holding more liquid assets, and the amount of the liquid asset also is decreasing over the time. Return on assets (ROA) indicates the profitability of a bank and an amount of operating incomes scaled by total assets.  From the result, we can observe non-systemic banks are more profitable.

\begin{figure}[htbp]
    \centering
    \caption{\textbf{Box Plots}\label{figboxplots}}
    \begin{subfigure}[b]{0.45\textwidth}
                \includegraphics[width=\textwidth]{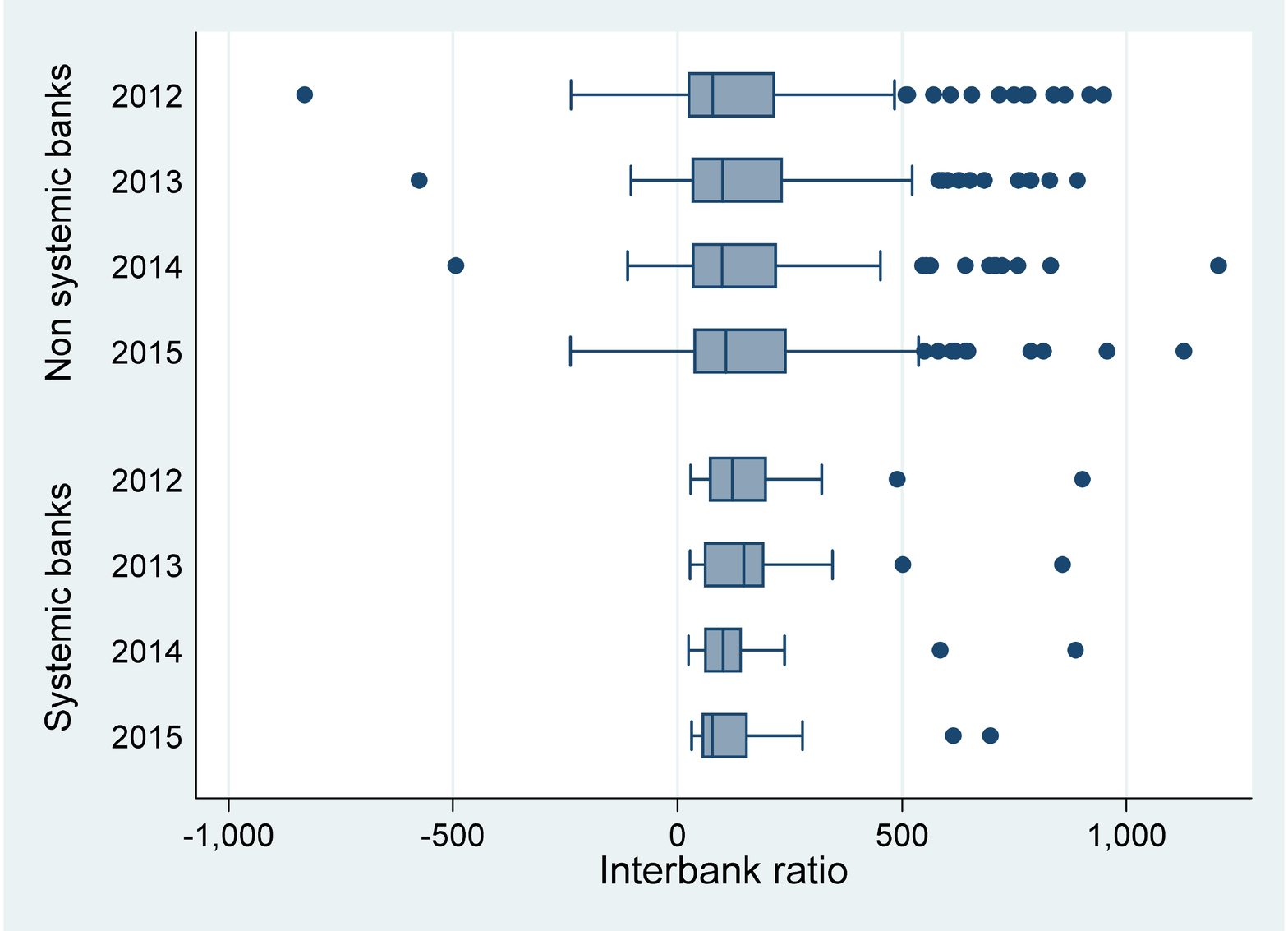}
               \caption{Box Plot of Interbank Ratios}
        \label{fig:gull}
    \end{subfigure}%
    \begin{subfigure}[b]{0.45\textwidth} \includegraphics[width=\textwidth]{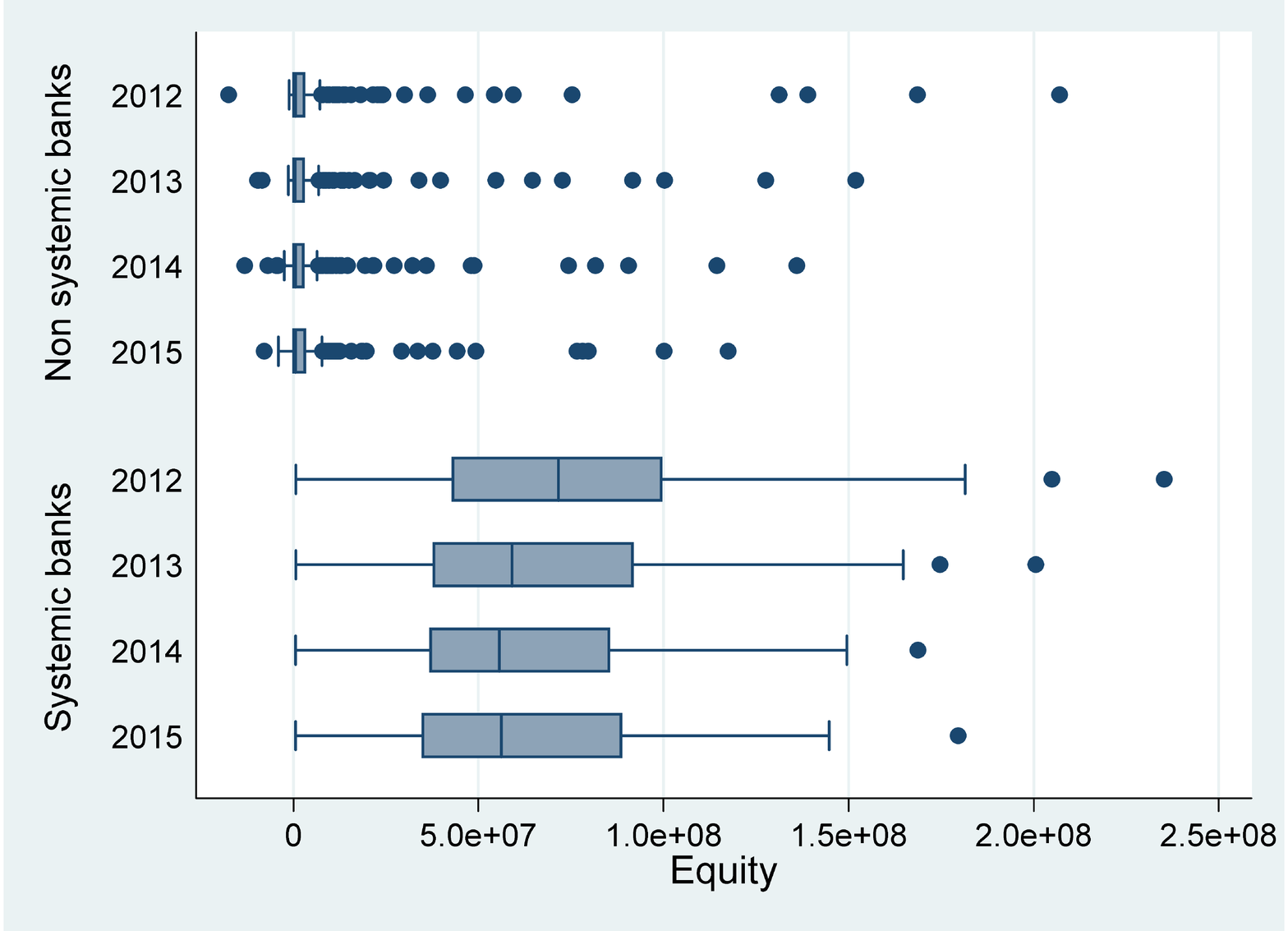}
             \caption{Box Plot of Equity }
    \end{subfigure}
    \begin{subfigure}[b]{0.45\textwidth} \includegraphics[width=\textwidth]{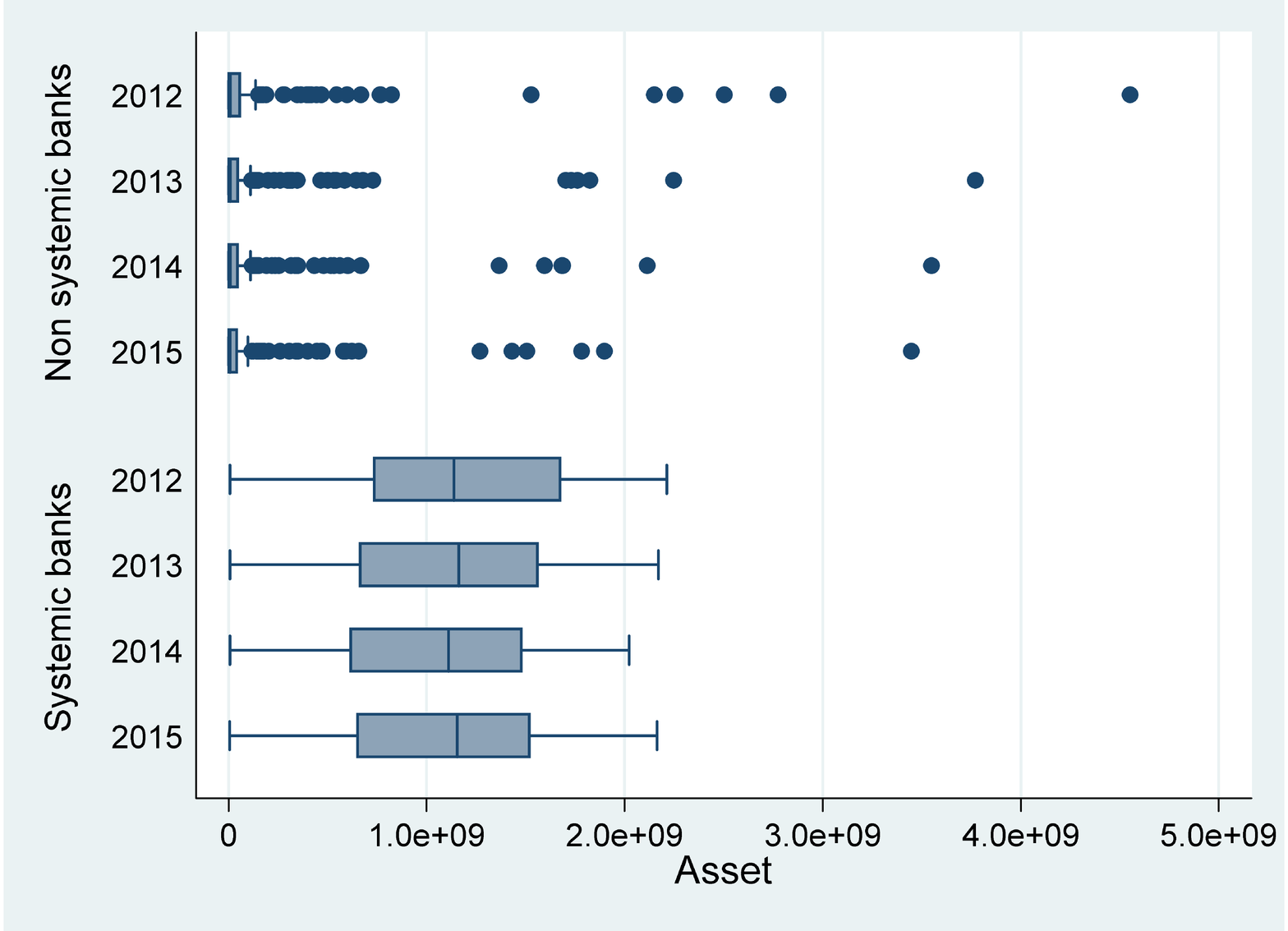}
\caption{Box Plot of Assets }
    \end{subfigure}    
    \captionsetup{font={small,it}}
\captionsetup{width=.8\textwidth}
\caption*{The figure plots the box plots (Box and Whisker Diagram )of the main variables (interbank, equity and asset)for systemic and non-systemic banks over years. The box shows the location of the middle 50\% of the data. It can be seen the difference between systemic and non-systemic banks and changing the amount of variables of systemic banks after recognizing as a G-SIFI. }
 \end{figure}

\begin{table}[htbp]
\centering
\def\sym#1{\ifmmode^{#1}\else\(^{#1}\)\fi}
\caption{\textbf{Sample Descriptive Statistics} \label{statdis}}

 \begin{subtable}[v]{1\textwidth}
        \centering
        \scalebox{0.8}{
        \begin{tabular} {@{} l l c c c c c c c c @{}} \\ \hline
 & Variable& \textbf{Mean} & \textbf{Std.Dev.} & \textbf{Obs.} & \textbf{Min} & \textbf{$25^{th}$ pct.} & \textbf{ Median} & \textbf{       $75^{th}$ pct.} & \textbf{       Max} \\
\hline\hline
Panel I: non systemic bank  & Asset&  1.42e+08 &   6.52e+08 &        772 &         41 &   968392.9 &    7080489 &   4.58e+07 &   1.40e+10
  \\
                  &  Ln(interbank ratio)& 4.384796 &   1.636508 &        733 &  -3.912023 &   3.708437 &   4.680092 &   5.483052 &   7.784854
  \\
                  &   Net loans/total assets & 34.46049 &   24.70773 &        772 &  -5.257587 &   8.432014 &     37.435 &    55.6685 &     91.676
  \\
                  &  Ln(deposits from banks) & 1.78e+07 &   7.18e+07 &        772 &  -4.72e+08 &      40660 &    1218596 &   1.22e+07 &   8.68e+08
  \\
                  & Ln (security for Sale) &  2.22e+07 &   3.37e+07 &        772 &  -3.74e+07 &   227321.5 &    5595969 &   3.47e+07 &   2.68e+08
  \\
                  & Liquid asset &   46.94603 &   62.15001 &        772 &  -24.25128 &     17.445 &   32.19779 &   60.05086 &     941.58
  \\
                  &  Cross-jurisdictional  & 73.61051 &   29.62508 &        546 &    12.0676 &    50.9147 &    68.8607 &   93.48219 &   155.6738
  \\
                  &  ROA &  1.138765 &   2.639387 &        772 &     -15.74 &       .275 &      1.095 &       2.08 &      19.33
  \\
                  &   Equity&  9540839 &   5.42e+07 &        772 &  -1.75e+07 &    58193.5 &   510518.5 &    2842683 &   7.30e+08
  \\
                      & &  &   &         &      &    &    &    &   
  \\
Panel II: systemic banks  & Asset&  1.10e+09 &   5.85e+08 &        112 &    5418100 &   6.61e+08 &   1.12e+09 &   1.55e+09 &   2.21e+09
  \\
                  &  Ln(interbank ratio)& 4.697136 &    .807599 &        112 &   3.205588 &   4.111053 &   4.729934 &   5.187104 &   6.804274
  \\
                  &  Net loans/total assets & 39.76857 &   16.11145 &        112 &      1.381 &     32.276 &    40.4515 &    51.6765 &     75.713
  \\
                  & Ln(deposits from banks) &  1.26e+08 &   8.92e+07 &        112 &      88100 &   5.25e+07 &   1.19e+08 &   1.92e+08 &   3.61e+08
  \\
                  & Ln (security for Sale) &  1.33e+08 &   1.30e+08 &        112 &  -946417.1 &   5.45e+07 &   8.36e+07 &   1.80e+08 &   6.82e+08
  \\
                  &  Liquid asset &  34.63473 &   19.22803 &        112 &       8.07 &      21.78 &     33.125 &      39.08 &       99.2
  \\
                  & Cross-jurisdictional  &  128.2727 &    10.2701 &        108 &   106.4944 &   122.6266 &   128.9281 &   136.2464 &   150.0981
  \\
                  &   ROA & .4790179 &   .4603601 &        112 &      -1.53 &        .26 &        .44 &       .725 &       1.51
  \\
                  & Equity&   7.19e+07 &   4.90e+07 &        112 &     520600 &   3.84e+07 &   5.91e+07 &   9.19e+07 &   2.35e+08
  \\
                    & &  &   &         &      &    &    &    &   
  \\
  
Panel III: total  & Asset&  2.64e+08 &   7.18e+08 &        884 &         41 &    1499488 &    9557034 &   1.21e+08 &   1.40e+10
  \\
                  &  Ln(interbank ratio)& 4.426195 &   1.555561 &        845 &  -3.912023 &   3.796001 &   4.686566 &   5.443074 &   7.784854
  \\
                  &  Net loans/total assets & 35.13301 &   23.84935 &        884 &  -5.257587 &   9.752098 &   38.31072 &    54.7755 &     91.676
  \\
                  & Ln(deposits from banks) &  3.15e+07 &   8.25e+07 &        884 &  -4.72e+08 &    68298.5 &    2078425 &   2.86e+07 &   8.68e+08
  \\
                  & Ln (security for Sale) &  3.62e+07 &   6.67e+07 &        884 &  -3.74e+07 &     325522 &   1.21e+07 &   4.84e+07 &   6.82e+08
  \\
                  &  Liquid asset &  45.38623 &     58.617 &        884 &  -24.25128 &      17.81 &     32.725 &   57.15964 &     941.58
  \\
                  &  Cross-jurisdictional  & 82.63729 &   34.09319 &        654 &    12.0676 &   54.26487 &   79.95827 &   110.8356 &   155.6738
  \\
                  &  ROA &  1.055178 &   2.481452 &        884 &     -15.74 &       .265 &       .935 &       1.91 &      19.33
  \\
                  &  Equity&  1.74e+07 &   5.74e+07 &        884 &  -1.75e+07 &    74176.5 &     815500 &    7435457 &   7.30e+08
  \\
\hline\hline
\end{tabular}
}
\end{subtable}
\begin{subtable}[v]{1\textwidth}
\captionsetup{justification=raggedleft,singlelinecheck=false}
\caption*{Tests of differences of mean of systemic and non systemic banks}
\centering
 \scalebox{0.8}{
\begin{tabular}{lrr}
\hline 
Variable & t-test & Wilcoxon-Mann-Whitney test \\ 
\hline \hline 
Asset & -15.9694 & -15.181 \\ 
Ln(interbank ratio) &  -3.2084 & -0.419\\ 
Net loans/total assets &  -3.0107& -1.779 \\ 
Ln(deposits from banks) & -21.1901 & -13.507 \\ 
Ln (security for Sale) & -22.1139 & -13.442 \\ 
Liquid asset & 4.2722 & 0.842 \\ 
Cross-jurisdictional &  -34.0047 & -14.753 \\ 
ROA &  6.3146 & 6.082\\ 
Equity &  -12.4121 & -15.409 \\ 
\hline \hline 
\end{tabular} }
\end{subtable}
\centering
\captionsetup{font={small,it}}
\captionsetup{width=0.9\textwidth}
 \caption*{The first table contains the information (provided by Bankscope) on bank-level for two classes of systemic and non-systemic banks. In this classification, non-systemic banks are corresponding to banks in the bucket 0, systemic banks in buckets 1-5 in the Table \ref{tebfrqSIFI}. Total asset is an indicator of the size of a bank and it contains the information on variances in asset liquidity. Interbank ratio measures the interbank liquidity and it is calculated as the proportion of due from counterparties to due to counterparties. This ratio with above one value indicates that the bank is a funds provider rather than a funds borrower.  Net Loans / Total Assets is a ratio of the assets of a bank are associated with loans.   A liquid asset is a total of cash, due from other banks like the central bank, the credit of Institutions as well as government and trading securities. Available‐for‐sale securities are the total fair value of securities which are available for sale and indicates unrealized profits or losses. Return on assets (ROA) is an operating income scaled by total assets and indicates profitability of a bank.\\
 The second table provides the results of t-test and the Wilcoxon-Mann-Whitney test of whether the mean for variables are the same for systemic and non-systemic banks. 
 Wilcoxon-Mann-Whitney test is a non-parametric test (distribution-free) analog with the t-test. }
\end{table}

\subsubsection{Empirical Results}

The section presents the results of our empirical study using bank panel data. Our potential caution of using the aggregated sample for all countries is  causing a bias in estimation of effects of factors to predict the systemic important banks. To relieve this issue, we restrict our data sample to information on only countries who are members of Economic Co-operation and Development (OECD). Then, we repeat our experience for all countries.  Our baseline specification related to identifying the factor influenced on the classification of a bank as a G-SIFI. Our dependent variable is an ordinal variable which contains information on global ranking of the banks. We regress this variable on the main indicators which determined by BCBS. We control banks profitability on the bank-level and GDP (PPP) per capital of the countries on the country-level. These variables measure the relative performance of banks and countries which the banks are located, respectively. Year fixed effect capture overall events and OECD-level effects control time-invariant market characteristics.\\

We empirically examine the prediction of our theoretical framework and test our hypothesis. Table \ref{TabMEALL} reports the results of identification systemic risk factors. The results show that the positive relation between indicators like interbank ratio, deposits from banks and Cross-jurisdictional on the estimation of G-SIFIs. Coefficients of logit regression are in log-odds units and the standard errors are reported in parentheses. In logistic regression, the odds the ratio is an association between an indicator and dependent variable and calculated as an exponential function of the logit coefficient. The results also show that some indicators like interbank ratio, deposits from banks and Cross-jurisdiction are associated with higher odds of systemic risk. Table \ref{TabMEALL} also reports the marginal effects of estimation the change in probability each bucket when the indicator variables rise by one unit. Although, decreasing the predictive power of the model due to few number of observations in the higher bucks, indicators like interbank ratio, deposits from banks and Cross-jurisdiction are still significant factors.

\begin{table}[htbp]\centering
\def\sym#1{\ifmmode^{#1}\else\(^{#1}\)\fi}
\caption{\textbf{Average Marginal Effects (ME) and Odds Ratios from Ordered Logit Regression for All Countries \label{TabMEALL}}}
        \scalebox{0.7}{
        
 \begin{tabular}{l*{7}{c}}
\hline\hline
            &\multicolumn{1}{c}{(1)}&\multicolumn{1}{c}{(2)}&\multicolumn{1}{c}{(3)}&\multicolumn{1}{c}{(4)}&\multicolumn{1}{c}{(5)}&\multicolumn{1}{c}{(6)}&\multicolumn{1}{c}{(7)}\\
            &\multicolumn{1}{c}{}&\multicolumn{1}{c}{Odds Ratios}&\multicolumn{1}{c}{Bucket 0}&\multicolumn{1}{c}{Bucket 1}&\multicolumn{1}{c}{Bucket 2}&\multicolumn{1}{c}{Bucket 3}&\multicolumn{1}{c}{Bucket 4}\\
\hline
Dependent variable: G-SIFI rank        &                     &                     &                     &                     &                     &                     &                     \\
Asset      &   -4.32e-10         &       1.000         &    2.06e-11         &   -2.26e-12         &   -6.44e-12         &   -6.18e-12         &   -5.75e-12         \\
            &  (4.39e-10)         &  (4.39e-10)         &  (2.11e-11)         &  (2.71e-12)         &  (6.75e-12)         &  (6.35e-12)         &  (6.03e-12)         \\
[1em]
Ln(interbank ratio)&       0.776\sym{***}&       2.172\sym{***}&     -0.0370\sym{***}&     0.00405         &      0.0116\sym{**} &      0.0111\sym{**} &      0.0103\sym{**} \\
            &     (0.209)         &     (0.454)         &   (0.00905)         &   (0.00216)         &   (0.00377)         &   (0.00391)         &   (0.00380)         \\
[1em]
 Net loans/total assets &      0.0141         &       1.014         &   -0.000672         &   0.0000735         &    0.000210         &    0.000202         &    0.000187         \\
            &   (0.00822)         &   (0.00834)         &  (0.000386)         & (0.0000559)         &  (0.000129)         &  (0.000125)         &  (0.000119)         \\
[1em]
Ln(deposits from banks)&       0.973\sym{***}&       2.645\sym{***}&     -0.0464\sym{***}&     0.00508         &      0.0145\sym{***}&      0.0139\sym{***}&      0.0129\sym{**} \\
            &     (0.206)         &     (0.546)         &   (0.00873)         &   (0.00284)         &   (0.00396)         &   (0.00413)         &   (0.00432)         \\
[1em]
Ln (security for Sale) &       0.354         &       1.424         &     -0.0169         &     0.00185         &     0.00527         &     0.00506         &     0.00471         \\
            &     (0.199)         &     (0.284)         &   (0.00937)         &   (0.00138)         &   (0.00307)         &   (0.00314)         &   (0.00289)         \\
[1em]
Liquid asset&     0.00401         &       1.004         &   -0.000192         &   0.0000209         &   0.0000598         &   0.0000574         &   0.0000534         \\
            &   (0.00284)         &   (0.00286)         &  (0.000136)         & (0.0000206)         & (0.0000433)         & (0.0000413)         & (0.0000401)         \\
[1em]
 Cross-jurisdictional&      0.0711\sym{***}&       1.074\sym{***}&    -0.00340\sym{***}&    0.000371         &     0.00106\sym{***}&     0.00102\sym{***}&    0.000947\sym{**} \\
            &    (0.0164)         &    (0.0176)         &  (0.000739)         &  (0.000228)         &  (0.000308)         &  (0.000299)         &  (0.000329)         \\
[1em]
ROA      &      -0.675\sym{***}&       0.509\sym{***}&      0.0322\sym{***}&    -0.00352         &     -0.0101\sym{***}&    -0.00966\sym{***}&    -0.00898\sym{**} \\
            &     (0.144)         &    (0.0734)         &   (0.00634)         &   (0.00209)         &   (0.00274)         &   (0.00284)         &   (0.00297)         \\
[1em]
Equity    &   -4.83e-09         &       1.000         &    2.31e-10         &   -2.52e-11         &   -7.20e-11         &   -6.92e-11         &   -6.43e-11         \\
            &  (2.76e-09)         &  (2.76e-09)         &  (1.29e-10)         &  (1.90e-11)         &  (4.19e-11)         &  (4.26e-11)         &  (4.03e-11)         \\
[1em]
Ln(GDPppp)  &       1.022\sym{***}&       2.777\sym{***}&     -0.0488\sym{***}&     0.00533         &      0.0152\sym{***}&      0.0146\sym{***}&      0.0136\sym{**} \\
            &     (0.205)         &     (0.568)         &   (0.00845)         &   (0.00310)         &   (0.00375)         &   (0.00406)         &   (0.00449)         \\
[1em]
Year Fixed      &      Yes         &    Yes         &  Yes        &   Yes         &   Yes         &     Yes    &     Yes            \\
\hline
cut1        &                     &                     &                     &                     &                     &                     &                     \\
\_cons      &       45.73\sym{***}& 7.22223e+19\sym{***}&                     &                     &                     &                     &                     \\
            &     (5.734)         &(4.14114e+20)         &                     &                     &                     &                     &                     \\
\hline
cut2        &                     &                     &                     &                     &                     &                     &                     \\
\_cons      &       47.92\sym{***}& 6.45478e+20\sym{***}&                     &                     &                     &                     &                     \\
            &     (5.840)         &(3.76945e+21)         &                     &                     &                     &                     &                     \\
\hline
cut3        &                     &                     &                     &                     &                     &                     &                     \\
\_cons      &       49.16\sym{***}& 2.23360e+21\sym{***}&                     &                     &                     &                     &                     \\
            &     (5.871)         &(1.31134e+22)         &                     &                     &                     &                     &                     \\
\hline
cut4        &                     &                     &                     &                     &                     &                     &                     \\
\_cons      &       50.36\sym{***}& 7.46518e+21\sym{***}&                     &                     &                     &                     &                     \\
            &     (5.883)         &(4.39189e+22)         &                     &                     &                     &                     &                     \\
\hline
\(N\)       &         581         &         581         &                  &                  &                 &                  &                  \\
Pseudo	$R^2$&           0.5024          &                     0.5024 &                     &                     &                     &                     &                     \\
\textit{AIC}&       433.3         &       433.3         &                    &                    &                    &                    &                    \\
\textit{BIC}&       503.1         &       503.1         &                    &                    &                    &                    &                    \\
\hline\hline
\multicolumn{8}{l}{\footnotesize Standard errors in parentheses}\\
\multicolumn{8}{l}{\footnotesize \sym{*} \(p<0.05\), \sym{**} \(p<0.01\), \sym{***} \(p<0.001\)}\\
\end{tabular}
}
    \captionsetup{font={small,it}}
\captionsetup{width=.8\textwidth}
\caption*{
This table reports the results of ordered logit regression model for estimation of main drivers of SIFI. The corresponding model is $GSIFI_{it}= const+X_{it}+C_{it}+\alpha$, where $X_{it}$ is the matrix explanatory variables and $C_{it}$ is the matrix control variables and $\alpha$ is fixed effect variable. The dependent variable is G-SIFI which contains the information on global ranking of banks. The second model is odds ratios of explanatory variables. If the odds ratio is above (below) one then the likelihood of probability of that event increase (decreases). We transfer skewed variables by taking the natural logarithm. Cuts are ancillary parameters which are applied to estimate Predicted probabilities of each bucket. The average marginal effects for each bucket is estimated. The Marginal effects explain the variation in probability when the explanatory variable increments by one unit.}
\end{table}

In table \ref{TabMEOECD}, we repeat our baseline model by restricting our sample to OECD countries only. Although the magnitude of the coefficients is slightly changed, we can observe consistent results with baseline model, and the estimated model for OECD subsample is comparable to the model for the entire sample.

\begin{table}[htbp]\centering
\def\sym#1{\ifmmode^{#1}\else\(^{#1}\)\fi}
\caption{\textbf{Average Marginal Effects (ME) and Odds Ratios from Ordered Logit Regression for OECD Countries only}\label{TabMEOECD}}
\scalebox{0.7}{
\begin{tabular}{l*{7}{c}}
\hline\hline
            &\multicolumn{1}{c}{(1)}&\multicolumn{1}{c}{(2)}&\multicolumn{1}{c}{(3)}&\multicolumn{1}{c}{(4)}&\multicolumn{1}{c}{(5)}&\multicolumn{1}{c}{(6)}&\multicolumn{1}{c}{(7)}\\
            &\multicolumn{1}{c}{}&\multicolumn{1}{c}{Odds Ratios}&\multicolumn{1}{c}{Bucket 0}&\multicolumn{1}{c}{Bucket 1}&\multicolumn{1}{c}{Bucket 2}&\multicolumn{1}{c}{Bucket 3}&\multicolumn{1}{c}{Bucket 4}\\
\hline
Dependent variable: G-SIFI rank        &                     &                     &                     &                     &                     &                     &                     \\
Asset      &    2.60e-10         &       1.000         &   -1.84e-11         &   -2.53e-12         &    6.92e-12         &    7.18e-12         &    6.82e-12         \\
            &  (2.87e-10)         &  (2.87e-10)         &  (2.01e-11)         &  (3.30e-12)         &  (7.86e-12)         &  (8.17e-12)         &  (7.56e-12)         \\
[1em]
Ln(interbank ratio)&       0.629\sym{**} &       1.876\sym{**} &     -0.0444\sym{**} &    -0.00613         &      0.0167\sym{*}  &      0.0174\sym{*}  &      0.0165\sym{*}  \\
            &     (0.218)         &     (0.408)         &    (0.0142)         &   (0.00448)         &   (0.00665)         &   (0.00707)         &   (0.00697)         \\
[1em]
Net loans/total assets &      0.0309\sym{**} &       1.031\sym{**} &    -0.00218\sym{***}&   -0.000301         &    0.000822\sym{*}  &    0.000852\sym{**} &    0.000809\sym{*}  \\
            &   (0.00990)         &    (0.0102)         &  (0.000657)         &  (0.000217)         &  (0.000334)         &  (0.000326)         &  (0.000316)         \\
[1em]
Ln(deposits from banks)&       0.735\sym{***}&       2.085\sym{***}&     -0.0519\sym{***}&    -0.00715         &      0.0195\sym{**} &      0.0203\sym{**} &      0.0192\sym{**} \\
            &     (0.210)         &     (0.437)         &    (0.0133)         &   (0.00481)         &   (0.00622)         &   (0.00718)         &   (0.00740)         \\
[1em]
Ln (security for Sale) &     -0.0134         &       0.987         &    0.000943         &    0.000130         &   -0.000355         &   -0.000368         &   -0.000350         \\
            &     (0.184)         &     (0.182)         &    (0.0130)         &   (0.00179)         &   (0.00492)         &   (0.00508)         &   (0.00482)         \\
[1em]
Liquid asset&     0.00462         &       1.005         &   -0.000326         &  -0.0000450         &    0.000123         &    0.000127         &    0.000121         \\
            &   (0.00325)         &   (0.00327)         &  (0.000231)         & (0.0000362)         & (0.0000901)         & (0.0000902)         & (0.0000898)         \\
[1em]
 Cross-jurisdictional&       0.122\sym{***}&       1.130\sym{***}&    -0.00862\sym{***}&    -0.00119         &     0.00325\sym{***}&     0.00337\sym{***}&     0.00320\sym{***}\\
            &    (0.0208)         &    (0.0235)         &   (0.00134)         &  (0.000676)         &  (0.000856)         &  (0.000867)         &  (0.000910)         \\
[1em]
ROA      &      -0.573\sym{**} &       0.564\sym{**} &      0.0405\sym{**} &     0.00558         &     -0.0153\sym{**} &     -0.0158\sym{*}  &     -0.0150\sym{*}  \\
            &     (0.197)         &     (0.111)         &    (0.0137)         &   (0.00345)         &   (0.00587)         &   (0.00625)         &   (0.00622)         \\
[1em]
Equity    &   -5.81e-09         &       1.000         &    4.10e-10         &    5.66e-11         &   -1.55e-10         &   -1.60e-10         &   -1.52e-10         \\
            &  (4.02e-09)         &  (4.02e-09)         &  (2.75e-10)         &  (5.84e-11)         &  (1.13e-10)         &  (1.20e-10)         &  (1.09e-10)         \\
[1em]
Ln(GDPppp)  &       0.684\sym{*}  &       1.982\sym{*}  &     -0.0483\sym{*}  &    -0.00666         &      0.0182\sym{*}  &      0.0189\sym{*}  &      0.0179\sym{*}  \\
            &     (0.283)         &     (0.560)         &    (0.0190)         &   (0.00502)         &   (0.00864)         &   (0.00847)         &   (0.00853)         \\
[1em]
Year Fixed      &      Yes         &    Yes         &  Yes        &   Yes         &   Yes         &     Yes       &     Yes     \\
\hline
cut1        &                     &                     &                     &                     &                     &                     &                     \\
\_cons      &       37.52\sym{***}& 1.96548e+16\sym{***}&                     &                     &                     &                     &                     \\
            &     (5.381)         &(1.05753e+17)         &                     &                     &                     &                     &                     \\
\hline
cut2        &                     &                     &                     &                     &                     &                     &                     \\
\_cons      &       40.28\sym{***}& 3.11060e+17\sym{***}&                     &                     &                     &                     &                     \\
            &     (5.544)         &(1.72455e+18)         &                     &                     &                     &                     &                     \\
\hline
cut3        &                     &                     &                     &                     &                     &                     &                     \\
\_cons      &       41.75\sym{***}& 1.35088e+18\sym{***}&                     &                     &                     &                     &                     \\
            &     (5.599)         &(7.56312e+18)         &                     &                     &                     &                     &                     \\
\hline
cut4        &                     &                     &                     &                     &                     &                     &                     \\
\_cons      &       43.08\sym{***}& 5.10301e+18\sym{***}&                     &                     &                     &                     &                     \\
            &     (5.624)         &(2.87006e+19)         &                     &                     &                     &                     &                     \\
\hline
\(N\)       &         258         &         258         &                  &               &               &                 &                  \\
Pseudo	$R^2$&          0.4541           &                    0.4541 &                     &                     &                     &                     &                     \\
\textit{AIC}&       354.8         &       354.8         &                    &                    &                    &                    &                    \\
\textit{BIC}&       411.7         &       411.7         &                    &                    &                   &                    &                    \\
\hline\hline
\multicolumn{8}{l}{\footnotesize Standard errors in parentheses}\\
\multicolumn{8}{l}{\footnotesize \sym{*} \(p<0.05\), \sym{**} \(p<0.01\), \sym{***} \(p<0.001\)}\\
\end{tabular}
}
    \captionsetup{font={small,it}}
\captionsetup{width=.8\textwidth}
\caption*{This table provides the results of baseline model for OECD countries only. Several econometric concerns arise in an international comparison a dataset contains information on all countries. The aim of OECD is to improve the economic and social well-being of words by setting international standards on various political and economic issues. Although banking regulations are varied across OECD countries, they are strongly connected and each country member can have significant influences  on other members. One can see that the sign and significant of coefficients of regression are consistent with baselines model and just magnitudes of coefficients are slightly have been changed.}
\end{table}

\subsubsection{Addressing Endogeneity Problem}

Generally, in estimating ordered categorical variables with using short panel data (with small T and $N \to \infty$) coefficients of logit and probit models might be biased (see : \cite{greene2003econometric}).  It creates bias since nuisance parameters which should be estimated grow rapidly as a number of panel (N) increases. Then data observed as grouped on the individual panel with few variation over time. In this condition, fixed effects models are mostly not estimable and regression with robust standard errors might not be applied.  The bootstrapped standard errors option can partially solve the problem. \cite{mccullagh1990simple} and \cite{firth1993bias} proposed a modification of log-likelihood estimation procedure to obtain a more stable estimation for short panel data. \cite{hayakawa2012gmm} proposed Generalized Method of Moments (GMM) estimators for the short dynamic panel. \\

In a short panel, errors are not correlated over panel, but over time for a given panel.  \cite{ mundlak1978pooling} proposed a correction method of logit model to control unobserved time-invariant individual heterogeneity. This method transforms the explanatory variables into cluster-mean deviations. This method adds either cluster-mean differenced variables or cluster-means to the list of explanatory variables (hybrid model). In other words, this method is a mean-centering procedure corresponding to moving the expectation of the original variable to the average point by subtracting off the variable averages over time from the data.\\

To apply this mean-centering operation, we use the following operator:

\begin{eqnarray}
\mathbf{X}^C = \mathcal{H}\mathbf{X},
\end{eqnarray}

where the operator $\mathcal{H}$ is defined as:
\begin{eqnarray}
 \mathcal{H} = (\mathbf{I}_{it}-\mathbf{1}_{.t}\mathbf{1}_{.t}^T).
\end{eqnarray}
The matrix-centering operator $ \mathcal{H}$ transfers the variable $\mathbf{X}$ to $\mathbf{X}^C$ which is a  mean-centered variable over time.\\

We consider the mentioned endogeneity problem and have a robust model we apply this correction transformation on the original variables. Table \ref{tabmundlak1978} reports the estimation of baseline model with correction of \cite{mundlak1978pooling} for the sample of all countries and a subsample which is restricted to only OECD Countries. One can see that the signs and significant of coefficients are similar to the baseline models, however the magnitude of the coefficients are changed especially in the first model. Since the correction method reduced the individual fixed effects which is correlated with error over time.\\

 Overall, these results support our main hypothesis which is the relation between the degree of interconnection and the likelihood of an institution to be classified as an SIFI.

\begin{table}[htbp]\centering
\def\sym#1{\ifmmode^{#1}\else\(^{#1}\)\fi}
\caption{\textbf{Ordered Logistic Regression for OECD and All Countries (Corrected \cite{mundlak1978pooling} )}\label{tabmundlak1978}}
        \scalebox{0.7}{
\begin{tabular}{lcccc}
\hline\hline
& (1)&(2)& (3) &(4)
\\
& OECD Countries& Odds Ratios (OECD Countries) & All Countries& Odds Ratios (All Countries)
\\
\hline
Dependent variable: G-SIFI rank  &                     &                     \\
$Asset^{C}$  &    6.10e-10         &       1.000    &    1.64e-11         &       1.000        \\
            &  (3.57e-10)         &  (3.57e-10)    &  (3.12e-10)         &  (3.12e-10)       \\
[1em]
$Ln(interbank)^{C}$&       0.260         &       1.297        &       0.698\sym{**} &       2.011\sym{**} \\
            &     (0.240)         &     (0.311)   &     (0.213)         &     (0.428)         \\
[1em]
$ Net loans/total assets^{C}$&-0.000000423\sym{*}  &       1.000\sym{*}  &-0.000000792\sym{***}&       1.000\sym{***}     \\
            &(0.000000205)         &(0.000000205) &(0.000000161)         &(0.000000161)        \\
[1em]
$Ln(deposits from banks)^{C}$&    1.68e-09         &       1.000        &    7.46e-09\sym{***}&       1.000\sym{***} \\
            &  (2.66e-09)         &  (2.66e-09)   &  (2.26e-09)         &  (2.26e-09)      \\
[1em]
$Ln (security for Sale)^{C}$ &    1.85e-09         &       1.000        &    3.87e-09\sym{**} &       1.000\sym{**} \\
            &  (1.48e-09)         &  (1.48e-09)   &  (1.46e-09)         &  (1.46e-09)        \\
[1em]
$Liquid asset^{C}$&     0.00169         &       1.002        &     0.00221         &       1.002         \\
            &   (0.00315)         &   (0.00315)     &   (0.00208)         &   (0.00208)       \\
[1em]
 $Cross-jurisdictional^{C}$&       0.129\sym{***}&       1.138\sym{***}&      0.0747\sym{***}&       1.078\sym{***}\\
            &    (0.0197)         &    (0.0224)  &    (0.0131)         &    (0.0142)       \\
[1em]
$ROA^{C}$ &      -0.599\sym{***}&       0.549\sym{***}&      -0.500\sym{***}&       0.607\sym{***}\\
            &     (0.155)         &    (0.0851)  &     (0.116)         &    (0.0704)           \\
[1em]
$Equity^{C}$&   -6.81e-09         &       1.000        &   -5.20e-09\sym{*}  &       1.000\sym{*}  \\
            &  (5.09e-09)         &  (5.09e-09)   &  (2.26e-09)         &  (2.26e-09)        \\
[1em]
$Ln(GDPppp)^{C}$  &       0.559\sym{*}  &       1.749\sym{*} &       0.967\sym{***}&       2.630\sym{***} \\
            &     (0.270)         &     (0.472)  &     (0.185)         &     (0.487)          \\
[1em]
Fixed year     &          Yes        &           Yes &           Yes         &           Yes         \\                        
\hline
cut1        &                    &          & &           \\
\_cons      &      -15.84         & 0.000000132   &      -69.20\sym{*}  &    8.88e-31\sym{*}        \\
            &     (35.74)         &(0.00000473)     &     (30.75)         &  (2.73e-29)        \\
\hline
cut2        &                &&     &                     \\
\_cons      &      -13.44         &  0.00000145    &      -67.16\sym{*}  &    6.80e-30\sym{*}        \\
            &     (35.73)         & (0.0000518)  &     (30.71)         &  (2.09e-28)           \\
\hline
cut3        &                     &      &&               \\
\_cons      &      -12.09         &  0.00000564  &      -65.95\sym{*}  &    2.29e-29\sym{*}          \\
            &     (35.73)         &  (0.000202)  &     (30.70)         &  (7.03e-28)         \\
\hline
cut4        &                     &    &&                 \\
\_cons      &      -10.82         &   0.0000201    &      -64.74\sym{*}  &    7.63e-29\sym{*}         \\
            &     (35.72)         &  (0.000717)   &     (30.68)         &  (2.34e-27)       \\
\hline
\(N\)       &         272         &         272        &         634         &         634      \\
Pseudo \(R^{2}\)& 0.413   &       0.413 & 0.4805&       0.4805 \\
\textit{AIC}&       385.2         &       385.2  &       459.8         &       459.8         \\
\textit{BIC}&       439.3         &       439.3   &       526.5         &       526.5       \\
\hline\hline
\multicolumn{3}{l}{\footnotesize Standard errors in parentheses}\\
\multicolumn{3}{l}{\footnotesize \sym{*} \(p<0.05\), \sym{**} \(p<0.01\), \sym{***} \(p<0.001\)}\\
\end{tabular}
}
    \captionsetup{font={small,it}}
\captionsetup{width=.8\textwidth}
\caption*{This table provides the results of the correction model of short panel  logit model proposed  by  \cite{ mundlak1978pooling}.  We transform the explanatory variables $X$ into cluster-mean deviations $X^C$. One can see that results are comparable to baselines model, excluding  interbank variable which is not significant factors in OECD model. It is because of the recent growing interconnectedness of OECD economies. The economic and banking activities of OECD are fragmented across countries.}
\end{table}

\section{Summary and Discussion}
\label{Secsummary}

In this paper, we propose a simple model to facilitate the understanding of systemic risk in financial systems. Our paper adds to the growing body of research on  identification of key players in a financial network by  introducing an index to rank institutions in a network. To achieve it, we formulated the problem of ranking institutions as a fixed point problem and reduced this problem into a convex optimization problem. We proposed an approximation solution to the problem and developed an algorithm for computing the fixed point to reach equilibrium in a static cascade model. The systemic rank of each institution is calculated by an impact index, and the result is robust to various exogenous shocks in the system. We also studied the underline distribution of this index and prove the existence and uniqueness of the index in a banking network. We applied our algorithm to a real payment system dataset and compared with an existing algorithm. \\

In our numerical experiments, we generated a random network and compared our approximation solution with an optimum solution.  Computational results suggest that our algorithm requires less CPU time than the optimum algorithm to exploit efficiently information about the systemically important nodes in a large-scale financial network. Results also indicate that CPU time does not increase exponentially when the number of nodes and interconnections of the network grow, which means that our algorithm has robust convergence behavior rate regardless of the network structure. \\

We also investigated the knife-edge property of the interbank system and test our hypothesis on increasing and decreasing the likelihood of an institution considered as an SIFI by applying some tools which were introduced by the BCBS.  The main objective of the experiment is to shed light on how to reduce the systemic impact of SIFIs by modifying or reformulating the regulations' policy. Numerical results show two main findings in the application of our methodology: firstly, it highlights the systemic aspects of interdependent relations among the institutions.  That is the mitigation of the systemic impact of SIFIs could be expressed by imposing restrictions on the level of their interconnections with other financial institutions. Secondly, our results provide an evidence that in a financial network the likelihood of a member to be considered as an SIFI increase and decrease by increasing the degree of outcomes (the number of debtors) and the degree of incomes (the number of creditors), respectively.\\

 We empirically investigate factors which lead to identifying global SIFIs by using a sample of the largest banks in each country. Our empirical outcomes support our theoretical finding that it is the association of the level of interconnection and the likelihood of an institution to be ranked as an SIFI.\\

This is a developing area of study in finance that still needs more experimentation. In practice, this type of analysis is not straightforward, mainly because of lack of sufficient knowledge about dynamics of the interbank system and inconsistently recorded data on policy interventions. Having access to the corresponding data allows us to map the entire financial system and assess its properties. There are several questions that arise from the reported results that are not beyond the scope of this research questions. We leave the study of the optimal policies and their impacts on financial stability in this framework to future works.

\newpage

\section{References}

\bibliographystyle{elsarticle-harv}
\bibliography{systemicRisk_single}

\end{document}